\documentclass[acmsmall]{acmart}

\setcopyright{none}
\acmJournal{POMACS}
\acmYear{2019} \acmVolume{3} \acmNumber{2} \acmArticle{30} \acmMonth{6} \acmPrice{15.00}\acmDOI{10.1145/3326145}

\renewcommand\footnotetextcopyrightpermission[1]{} 
\usepackage{xpatch}
\makeatletter
\xpatchcmd{\ps@firstpagestyle}{\@journalNameShort, Vol. \@acmVolume, No.
      \@acmNumber, Article \@acmArticle.  Publication date:
      \@acmPubDate.}{}{\typeout{First patch succeeded}}{\typeout{first patch failed}}
\makeatother

\usepackage{amsmath,amssymb,amsfonts,amsthm}
\usepackage{graphicx}
\usepackage{textcomp}
\usepackage{xcolor}
\usepackage{array}
\usepackage{booktabs}
\usepackage{pifont}
\usepackage{subfigure}
\usepackage{url}
\usepackage{enumitem}
\usepackage{soul}

\newtheorem{theorem}{Theorem}

\newtheorem{definition}{Definition}
\newtheorem{lemma}{Lemma}
\newtheorem{problem}{Problem}

\usepackage{algorithm}
\usepackage{algorithmicx}
\usepackage{algpseudocode}

\usepackage{comment}
\usepackage{enumitem}

\usepackage[normalem]{ulem}

\definecolor{purple}{rgb}{1,0,1}
\definecolor{red}{rgb}{0,0,0}
\definecolor{blue}{rgb}{0,0,0}
\definecolor{camera}{rgb}{0,0,0}
\definecolor{camerax}{rgb}{0,0,0}
\let\emptyset\varnothing
\newcommand{\red}[1]{\textcolor{camera}{#1}}
\newcommand{\blue}[1]{\textcolor{camera}{#1}}
\newcommand{\camera}[1]{\textcolor{camera}{#1}}
\newcommand{\camerax}[1]{\textcolor{camerax}{#1}}
\newcommand{\todo}[1]{\textcolor{red}{TODO: #1}}
\newcommand{\pavlos}[1]{\textcolor{purple}{\textbf{PS:} #1}} 
\newcommand{\vkotronis}[1]{\textcolor{green}{\textbf{VK:} #1}} 

\newcommand{\eg}{\textit{e.g., }}
\newcommand{\ie}{\textit{i.e., }}
\newcommand{\myitem}[1]{\vspace*{0.04in}\noindent\textbf{#1}}
\newcommand{\myitemit}[1]{\vspace*{0.04in}\noindent\textbf{\textit{#1}}}
\newcommand{\aspath}[2]{p_{#1 \rightarrow #2}}
\newcommand{\bestaspath}[2]{bp_{#1 \rightarrow #2}}
\newcommand{\eq}[1]{Eq.~(\ref{#1})}
\newcommand{\LineComment}[1]{\hfill/* \textit{#1} */}
\newcommand{\yes}{$\checkmark$}

\usepackage{xspace}
\newcommand{\Rgraph}{R-graph\xspace}
\newcommand{\Rpath}{R-path\xspace}

\newcommand{\ppoint}{ingress point\xspace}
\newcommand{\ppoints}{ingress points\xspace}

\begin{document}

\title{Inferring Catchment in Internet Routing}

\thanks{\camera{This work has been funded by the European Research Council grant agreement no. 338402.}}

\author{Pavlos Sermpezis}
\affiliation{%
  \institution{Institute of Computer Science, FORTH}
  \city{Heraklion}
  \country{Greece}
}
\email{sermpezis@ics.forth.gr}

\author{Vasileios Kotronis}
\affiliation{%
  \institution{Institute of Computer Science, FORTH}
  \city{Heraklion}
  \country{Greece}
}
\email{vkotronis@ics.forth.gr}

\renewcommand{\shortauthors}{Pavlos Sermpezis and Vasileios Kotronis}


\begin{abstract}
BGP is the de-facto Internet routing protocol for exchanging prefix reachability information between Autonomous Systems (AS). It is a dynamic, distributed, path-vector protocol that enables rich expressions of network policies (typically treated as secrets). In this regime, where complexity is interwoven with information hiding, answering questions such as ``what is the expected catchment of the anycast sites of a content provider on the AS-level, if new sites are deployed?'', or ``how will load-balancing behave if an ISP changes its routing policy for a prefix?'', is a hard challenge. In this work, we present a formal model and methodology that takes into account policy-based routing and topological properties of the Internet graph, to predict the routing behavior of networks. We design algorithms that provide new capabilities for \textit{informative} route inference (e.g., isolating the effect of randomness that is present in prior simulation-based approaches).
We analyze the properties of these inference algorithms, and evaluate them using publicly available routing datasets and real-world experiments. The proposed framework can be useful in a number of applications: measurements, traffic engineering, network planning, Internet routing models, etc. As a use case, we study the problem of selecting a set of measurement vantage points to maximize route inference. Our methodology is general and can capture standard valley-free routing, as well as more complex topological and routing setups appearing in practice.
\end{abstract}



\maketitle

\vspace{5\baselineskip}
\section{Introduction}\label{sec:introduction}
Routing between networks (or Autonomous Systems--AS) in the Internet takes place via the Border Gateway Protocol (BGP)~\cite{bgprfc4271}. BGP is a policy-based, destination-oriented path-vector protocol, where an AS receives paths to a destination network from its neighbors, selects which path to prefer based on its local routing policies, and advertises it to other neighbors based on its export policies. This typically results in asymmetric paths between networks~\cite{giotsas2014inferring}. Each destination network has control only over \emph{its own} routing decisions, and typically \textit{cannot~control or even know how other networks route their traffic} to it.

Knowing how networks route traffic to a destination is important for (i) network planning or monitoring (\eg allocation of network resources, detection of routing anomalies)~\cite{cicalese2015fistful,de2017broad,gursun2012traffic}, and (ii) indirect control --if possible-- of routing decisions of other networks (\eg through manipulation of BGP announcements
, selection of local routing policies, establishment of new links)~\cite{lodhi2015complexities,rfc4786}. Specifically, for a destination network, it is of particular interest to know from which of its \textit{\ppoints} (\eg border routers) it should expect to receive traffic from other networks  \blue{under a given routing configuration}~\cite{baltra2014ingress}. \blue{We consider the following indicative examples.}

\noindent\textit{Example A:} A regional ISP $R$, whose network spans a region of two major cities, $city_{A}$ and $city_{B}$, has a single upstream tier-1 ISP $T_{A}$ and connects to it at $city_{A}$. To avoid overloading its infrastructure in $city_{A}$, $R$ decides to connect to another tier-1 ISP $T_{B}$ at $city_{B}$. However, after connecting with $T_{B}$, $R$ observes that $90\%$ of the incoming Internet traffic still enters its network at $city_{A}$, therefore the new setup fails to balance $R$'s load among its infrastructure in the two cities. \blue{In fact, how to select a transit provider, is a question that lacks a clear answer, and engages operators in active discussions}~\cite{nanog-select-transit}.

\noindent\textit{Example B:} A content provider $C$ applies IP anycast (\ie announces the same IP prefix)~\cite{rfc4786,cicalese2015fistful,moura2016anycast,de2017broad,ra2017seeing} from three sites. Due to traffic increase, $C$ decides to add one more anycast site. 
It needs to select where to deploy and how to connect the new site, in order to best split the traffic among its sites. \blue{The ongoing research in IP anycasting, \eg \cite{de2017broad,ra2017seeing,anycast-sigcomm2018}, indicates that this is a problem that is not well-understood yet.} 

While a network can \emph{partially} determine how other networks route traffic to it through passive (\eg BGP data~\cite{caidabgpstream}) or active (\eg traceroute, ping) measurements~\cite{baltra2014ingress,mao2005level,lee2011scalable,cicalese2015fistful,de2017broad,ra2017seeing}, \textit{measurements can provide information only for an existing deployment}. However, in many applications (traffic engineering, peering decisions, network resilience, etc.)~\cite{lodhi2015complexities}, it is important to know, \ie predict, how the routing behavior of other networks will change \textit{in advance}, before a network actually alters its local policies or connections. 
\blue{Moreover, even when it is possible to afford several trials to test different traffic engineering (TE) decisions, the large number of possible options limits the efficiency and/or applicability of this ``trial-and-error'' approach, \textit{unless an informed methodology is used}.}

To this end, \blue{the primary goal of this paper is to \textit{provide an informative inference or prediction for}} \textit{the catchment of the \ppoints of a network, under a given (existing or not) topological and routing configuration}. With the term ``catchment'' (see, \eg~\cite{rfc4786,de2017broad,ra2017seeing}) of an \ppoint $m$ of a destination network $n_{dst}$, we denote the set of networks that route their traffic to $n_{dst}$ through  $m$.

Route inference is identified as a challenging task~\cite{rfc4786,lodhi2015complexities}, due to the inherent complexity of the behavior of BGP mechanisms, and lack of public data for networks' routing policies (in fact, only coarse estimates are available, \eg the AS-relationships~\cite{luckie2013relationships, caidaasrel}). \blue{Moreover, the related problems (\eg TE optimization) that may
arise in practice are typically of combinatorial nature~\cite{muhlbauer2006building}.}

\blue{The common approach to predict routes is to use models, such as the valley-free model~\cite{gao2001stable} or other variants~\cite{c-bgp,muhlbauer2006building,feamster2004model}, that simulate the Internet routing process (BGP) based on available data. Policies are typically inferred from public data~\cite{caidaasrel,muhlbauer2006building}, and when there is lack of, or coarse-grained, knowledge of policies, they are arbitrarily selected (\eg random tie-breaking), in order to proceed to a simulation and obtain a prediction. However, a simulator computes only one of all the possible outcomes per simulation run. Thus, this approach can lead to an \textit{output that is heavily affected by the introduced randomness} 
(\eg breaking randomly a tie for a central AS in each simulation, may lead to high catchment for an \ppoint $m$ in one simulation run, and to low catchment in another run).
Most importantly, the output \textit{does not reveal what is the effect of the randomness}, \eg how many routes are affected by an arbitrarily chosen policy.}

\blue{
In this paper, we revisit the problem of route/catchment inference, and propose a
framework and methodology for an informative inference that 
quantifies the certainty/uncertainty in the prediction for every network (isolating the effect of randomness), and reveals the factors that affect the inference (\eg certain policies or networks). This in turn enables the development of advanced methods for optimizing traffic engineering, selecting peering strategies, or conducting measurement campaigns. Specifically:}
\begin{itemize}[leftmargin=*,noitemsep,topsep=0pt]
\item We formally model (Section~\ref{sec:model}) and study the problem of inferring the catchment of the \ppoints of a network. \blue{To this end, we propose a graph structure, the \Rgraph, that can efficiently encode rich information about the routing behavior, and isolate the effect of randomness  (Section~\ref{sec:build-vf-graph}).}

\item \blue{We design and analytically study methodologies
that infer catchment in existing or hypothetical scenarios (Section~\ref{sec:methodology}). 
We
identify the networks for which a certain inference is possible, even under coarse estimates of routing policies and topology (Section~\ref{sec:vf-graph-based-inference}), calculate the \textit{probabilistic inference} for the remaining networks (Section~\ref{sec:probabilistic-inference}), and 
enhance the inference when some oracles (\eg from \textit{measurements}) are given (Section~\ref{sec:measurement-enhanced-inference}).
}
\camera{The code for an implementation of the proposed methods is available in~\cite{sermpezis2019code}.}

\item \blue{As a use case of our framework, we consider and study the problem of maximizing the inference of catchment under a limited budget of measurements. We propose an efficient greedy algorithm, which leverages the structure of the \Rgraph, for selecting the measurement targets (Section~\ref{sec:measurements}). \camera{Our analysis sheds light on the complexity of problems related to route inference, and can be of more general interest.}}

\end{itemize}
\blue{While the main focus of the paper is on establishing a theoretical framework for catchment inference, we provide an initial evaluation of} 
our approach in realistic Internet routing scenarios through extensive simulations \camera{and real experiments}, and provide insightful results (Section~\ref{sec:results}). \blue{We present related work in Section~\ref{sec:related}, and conclude \camera{by discussing potential applications and future research directions} in Section~\ref{sec:conclusions-future-work}.}

\blue{As a final remark, we would like to stress that our goal is not to propose a new inter-domain routing model~\cite{gao2001stable}, or infer more accurately the routing policies in the Internet~\cite{muhlbauer2006building}, but to provide inference methods and insights on top of \textit{any} given model and set of policies. Finally, we believe that our methods can be useful for more general applications of BGP (or, similar policy-oriented path-vector protocols), apart from inter-domain routing, such as in iBGP or data centers~\cite{bgp-in-dc-rfc7938}
}. 

\section{Model}\label{sec:model}
We present our model in Section~\ref{sec:generic-model}, and provide the necessary definitions related to route inference in Section~\ref{sec:ppoints}. In Section~\ref{sec:vf-model} we discuss how the commonly used valley-free model~\cite{gao2001stable} can be captured as a special case of our generic model. Important notation is summarized in Table~\ref{table:notation}.
\subsection{Network and Routing}\label{sec:generic-model}
We assume a network with a set of nodes $\mathcal{N}$ and edges $\mathcal{E}$. 
\blue{A node may correspond to a single AS, or a part of an AS (\eg in case of large/distributed ASes; similarly to the concept of ``quasi-routers'' in~\cite{muhlbauer2006building}), or even a group of ASes with the same routing policies (\eg siblings). For brevity, and without loss of generality, in the remainder we consider that a node represents a single AS\footnote{The study of \cite{muhlbauer2006building} showed that more than 98\% of the ASes can be accurately (w.r.t. inter-domain routing behavior) represented as a single node / 
``quasi-router''.}, and an edge corresponds to a peering link between two ASes.}
We refer to the nodes connected with an edge to a node $i$, as the neighbors of $i$.

\myitem{Routing protocol and policies.} Nodes use BGP~\cite{bgprfc4271}
to establish routes towards different Internet destinations. The main operation of BGP is described as follows. A destination node $n_{dst}$ announces a prefix. Every other node $i$ learns from its neighbors paths to $n_{dst}$ (\ie its prefix), stores them in a local routing table (Routing Information Base, RIB), and selects one of them as its best path to $n_{dst}$ (according to, \eg its \textit{local preferences}). Then, $i$ may advertise this best path to its neighbors (according to its \textit{export policies}).

A path contains a sequence of nodes; we denote a path from $i$ to $j$ as $\aspath{i}{j}$, and use the following notation:
\begin{equation*}
\aspath{i}{j}=[i,x,y,...,z,j]
,~~~i,x,y,...,z,j\in\mathcal{N}
\end{equation*}
We further denote the best path from $i$ to $j$ (\ie the path that $i$ prefers --among all paths in its RIB-- to reach $j$) as $\bestaspath{i}{j}$.

\textit{Best path selection}. Each node $i$ assigns a \textit{local preference} to each of its neighbors. We denote the set of local preferences in the network as $\mathcal{Q} = \{q_{ij}\in\mathbb{R}: i,j\in\mathcal{N},e_{ij}\in\mathcal{E}\}$. Note that local preferences are in general asymmetric, \ie $q_{ij}\neq q_{ji}$. If paths are learned from more than one neighbors, then $i$ prefers the path learned from the neighbor with the highest local preference~\cite{bgprfc4271}. If a node $i$ has the same local preference for two neighbors $j$ and $k$ ($q_{ij}=q_{ik}$), then the selection is based on other criteria (``tie-breakers''), such as path length (see Section~\ref{sec:weak-path-inference}), the MED attribute, IGP metrics, time of advertisement, etc.~\cite{ciscobestpath}.

\textit{Path export.} When a node $i$ selects a best path for $n_{dst}$ via a neighbor $j$, it may advertise (export) this path to all, some or none of its neighbors. We denote the set of export policies as $\mathcal{H} = \{h_{ijk}\in\{0,1\}: i,j,k\in\mathcal{N},e_{ij},e_{ik}\in\mathcal{E}\}$, where $h_{ijk}=1$ denotes that $i$ exports to $k$ a path learned from $j$ (and $h_{ijk}=0$ otherwise). Typically, both export policies and local preferences are based on the economic relationships between the nodes, and are consistent with each other.
Therefore, it is safe to assume for practical setups that $q_{ij} = q_{i\ell}\Rightarrow h_{ijk} = h_{i\ell k}$, $\forall k$, \ie
routes learned from neighbors with the same local preference are similarly treated\footnote{\blue{In case a node has different export policies for neighbors of same local preference, we can split the node into more than one sub-nodes (with the same neighbors and local preferences), each of them corresponding to one export policy.}}.

\myitem{}
\blue{
\textit{Remark on the generality of the model:} (i) The proposed model allows to capture generic routing policies by carefully selecting the quantities $\mathcal{Q}$ and $\mathcal{H}$; even sophisticated per-prefix policies can be captured by considering different policies $\mathcal{Q}^{p}$ and $\mathcal{H}^{p}$ per prefix $p$. (ii) The model can be applied in generic settings: when the detailed policies of a node are known by explicitly setting the $\mathcal{Q}$ and $\mathcal{H}$ values; or when we have only coarse-grained information about them (see Section~\ref{sec:vf-model}); or even when we entirely lack policy information for some nodes, where its values for $\mathcal{Q}$ and $\mathcal{H}$ can be set equal to a default value, thus without excluding any possible outcome.
}

\myitem{Eligible paths.} We define the \textit{eligible paths} of a node $i$ to a node $n_{dst}$, as the paths that can be in the RIB of $i$; thus, one of them can be selected by $i$ as its best path to $n_{dst}$. The eligible paths are later used in the route inference methodology (Section~\ref{sec:vf-graph-based-inference}).
\begin{definition}[Eligible path]\label{def:eligible-path}
An \emph{eligible path} $\aspath{i}{n_{dst}}$ is a path from $i$ to $n_{dst}$ that (i) conforms to the routing policies $\mathcal{Q}$ and $\mathcal{H}$, and (ii) can be selected by $i$ as its best path to $n_{dst}$.
\end{definition}
The first condition in Def.~\ref{def:eligible-path} dictates that only paths that can be received by $i$ (\ie be in its RIB) can be eligible. For example, if $h_{ijk}=0$, then $i$ will not export to $k$ a path learned from $j$, and thus the path $[k,i,j,...,n_{dst}]$ does not conform to the routing policies $\mathcal{H}$ and cannot be eligible.
The second condition excludes paths which are not preferred due to $i$'s local preferences. For instance, let $i$ have in its RIB two paths $\aspath{i}{n_{dst}}^{(1)} = [i,j,...,n_{dst}]$ and $\aspath{i}{n_{dst}}^{(2)} = [i,k,...,n_{dst}]$.
If $q_{ij}>q_{ik}$, then $\aspath{i}{n_{dst}}^{(2)}$ is \textit{not} eligible, since $\aspath{i}{n_{dst}}^{(1)}$ will always be preferred by $i$. However, if $q_{ij}=q_{ik}$, both paths are eligible.

\begin{table}
\centering
\caption{Important Notation.}\label{table:notation}
\vspace{-2.5mm}
\begin{small}
\begin{tabular}{|l|l|}
\hline
{$\mathcal{G}$} & {Network graph; $\mathcal{G}(\mathcal{N},\mathcal{E}, \mathcal{Q},\mathcal{H})$}\\
\hline
{$\mathcal{N}$} & {Set of nodes in $\mathcal{G}$}\\
\hline
{$\mathcal{E}$} & {Set of edges $e_{ij}$ in $\mathcal{G}$}\\
\hline
{$\mathcal{Q}$} & {Local preferences $\mathcal{Q} = \{q_{ij}\in\mathbb{R}: i,j\in\mathcal{N},e_{ij}\in\mathcal{E}\}$}\\
\hline
{$\mathcal{H}$} & {Export policies   }\\
{}&{$\mathcal{H} =  \{h_{ijk}\in\{0,1\}: i,j,k\in\mathcal{N},e_{ij},e_{ik}\in\mathcal{E}\}$}\\
\hline
{$\aspath{i}{j}$}&{Path from node $i$ to node $j$}\\
\hline
{$\bestaspath{i}{j}$}&{Best path from node $i$ to node $j$}\\
\hline
{$\mathcal{M}$} & {Ingress points of the destination node $n_{dst}$}\\
\hline
{$i\rhd m$} & {Node route; $i$ reaches $n_{dst}$ through the \ppoint $m$}\\
\hline
{$\pi_{i}(m)$} & {Route probability for node $i$ and \ppoint $m$, \eq{eq:definition-route-probability}}\\
\hline
{$f$} & {Routing function, \eq{eq:routing-function}}\\
\hline
{$\mathcal{G}_{R}$} & {\Rgraph; $\mathcal{G}_{R}(\mathcal{N}_{R},\mathcal{E}_{R})$}\\
\hline
\end{tabular}
\end{small}
\end{table}

\subsection{Ingress Points and Catchment}\label{sec:ppoints}
\myitem{Ingress points.} Let a network $n_{dst}$ that originates a prefix, and is connected to its neighbors (and receives traffic) through a set of \ppoints $\mathcal{M}$. An \ppoint can be a router interface of $n_{dst}$ that is used exclusively in a private peering link (\eg with its upstream provider) or a router at an IXP that connects $n_{dst}$ to multiple other networks (\ie the members of the IXP). 

\textit{Remark:} This notion can be generalized for the case where multiple nodes announce the same prefix (Multi-Origin AS, or MOAS). A virtual node $n_{dst}$ can be connected to these MOAS nodes, which then serve as the \ppoints of $n_{dst}$.

\myitem{Catchment: mapping nodes to \ppoints.} Let assume w.l.o.g. that each neighbor $j$ of $n_{dst}$ is directly connected to $n_{dst}$ through exactly one \ppoint $m$, $m\in\mathcal{M}$. We denote this as $j\rhd m$.
Every other node $i$, $i\in\mathcal{N}$, selects a best path $\bestaspath{i}{n_{dst}}$ towards $n_{dst}$, \eg 
\begin{equation*}
\bestaspath{i}{n_{dst}} = [i,x,...,y,n_{dst}]
\end{equation*}
where $x$ is a neighbor of $i$, and $y$ a neighbor of $n_{dst}$. In this example, if $y\rhd m$, $m\in\mathcal{M}$, then $\bestaspath{i}{n_{dst}}\rhd m$ and $i\rhd m$.

\begin{definition}[Node route / Catchment]\label{def:route-catchment}~

\noindent The \underline{route} of a node $i$ is $m$, and is denoted as $i\rhd m$, when $i$ routes its traffic to $n_{dst}$ through the \ppoint $m$ of $n_{dst}$.

\noindent The \underline{catchment} of an \ppoint $m$ is the set of nodes $i\in\mathcal{N}$, for which it holds that $i\rhd m$.
\end{definition}

We would like to stress that the ``route'' of a node $i$, as defined in Def.~\ref{def:route-catchment} and used throughout the paper, indicates only how the traffic of $i$ enters the network of $n_{dst}$ (\ie the last hop closest to $n_{dst}$ in $\bestaspath{i}{n_{dst}}$), and not the entire AS-path.

\myitem{Route probability and Routing function.} 
In many cases we cannot determine which is the best path of a node $i$, \eg when  the paths $\aspath{i}{n_{dst}}$ (i) are not known, or (ii) are known but the local preferences are unknown. We capture this uncertainty in a probabilistic way, by defining the \textit{route probability} as:
\begin{equation}\label{eq:definition-route-probability}
\pi_{i}(m) = Prob\{\bestaspath{i}{n_{dst}}\rhd m\},~~~~i\in\mathcal{N},m\in\mathcal{M}
\end{equation}
Furthermore, we define the \textit{routing function} $f:\mathcal{N}\rightarrow \mathcal{M}\cup \{0\}$ that maps nodes ($i\in\mathcal{N}$) to \ppoints ($m\in\mathcal{M}$) as:
\begin{equation}\label{eq:routing-function}
f(i) = \left\{
\begin{tabular}{ll}
{m} & {, if  $\pi_{i}(m)=1$ }\\
{0} & {, otherwise }
\end{tabular}
\right.
\end{equation}
In other words, $f(i)=m\neq 0$ denotes a certainty for the route of node $i$ (and $f(i)=0$ denotes uncertainty).

\subsection{A Sub-Case: the Valley-Free (VF) Model}\label{sec:vf-model}


The network and routing model of Section~\ref{sec:generic-model} are generic and can describe the BGP setups encountered in practice.
\blue{Here, we present how the valley-free (VF) routing model~\cite{gao2001stable} can be captured as a special case of our~model. The VF model is widely considered in related work as a useful approximation for Internet routing, thus we believe that this section will facilitate other researchers to apply our framework.}

In the VF model, each pair of adjacent nodes has either a \textit{customer-to-provider} or a \textit{peer-to-peer} relationship%
%
%
. We denote a relationship between two nodes $i,j$ ($i,j\in\mathcal{N},e_{ij}\in\mathcal{E}$) as $\ell_{ij}\in \{\textit{c2p,p2p,p2c}\}$, \eg $\ell_{ij}=\textit{c2p}$ when $i$ is a customer of $j$.
Note that when $\ell_{ij}=c2p$ then $\ell_{ji}=p2c$, but \textit{p2p} relationships are typically symmetric (e.g., settlement-free peering).

Under the VF model, a node $i$ prefers paths received from customers to paths from peers or providers, and paths from peers to paths from providers. We denote this path preference as $p2c \succ p2p \succ c2p$
and we can capture this in our model by assigning local preferences as follows:
\begin{equation}\label{eq:vf-local-preferences}
q_{ij}>q_{ik}~~~\Leftrightarrow~~~\ell_{ij}\succ\ell_{ik}
\end{equation}

Moreover, when a node has a best path for $n_{dst}$ through a customer, it advertises this path to all its neighbors (customers, peers, providers); and when the best path is through a peer or provider, it advertises this path only to its customers:

\begin{equation}\label{eq:vf-export-policies}
h_{ijk}=
\begin{cases} 
      1 & , ~\text{if}~~~ \ell_{ik}=p2c ~\text{or}~ \ell_{ij}=p2c \\
      0 & , ~\text{otherwise}~ 
   \end{cases}
\end{equation}

It is worth noting that in practice, only coarse estimates of the AS-relationships $\ell_{ij}$ are known (\eg CAIDA AS-relationship dataset~\cite{caidaasrel}), while the detailed local preferences $q_{ij}$ are typically not made public by networks. Hence, it is commonly assumed that $q_{ij}=q_{ik}\Leftrightarrow\ell_{ij}=\ell_{ik}$, \ie a network assigns equal local preferences to all neighbors of the same type~\cite{caesar2005bgp}.

\section{Route Inference}
\label{sec:methodology}
\myitem{The problem.} Our goal is to infer through which \ppoint each node $i$ reaches the destination node $n_{dst}$ (or, equivalently, the \textit{route} of each node / the catchment of each \ppoint).
\blue{In this section, we tackle this problem, and provide methods for the route inference. Our methodology is summarized as follows.}

\myitem{Methodology overview.} 
We first calculate for every node $i\in\mathcal{N}$ all its eligible paths to $n_{dst}$ (see Def.~\ref{def:eligible-path}), and encode them in a directed acyclic graph (DAG) rooted in $n_{dst}$; we call this graph the \textit{Routing Graph} or \textit{\Rgraph} (Section~\ref{sec:build-vf-graph}). \blue{The \Rgraph is the basic structure, on which our inference methodology is built.}

\blue{Proceeding to inference,} we first focus on the nodes for which a \textit{certain} inference can be made (Section~\ref{sec:vf-graph-based-inference}); \textit{our goal is to calculate} $f(i)$, $\forall i\in\mathcal{N}$. We infer the values of the routing function $f$ based on the structure of the \Rgraph; when $i$ has only one eligible path $\aspath{i}{n_{dst}}$, and this path is through the \ppoint $m$, then $f(i)=m$. \blue{However, and most importantly}, the \Rgraph enables 
to determine non-zero values of $f(i)$ (\ie certain inference) 
also for some nodes that have multiple eligible paths; even without knowing which of them is the best path, or enumerating all of them.

\blue{We then focus on nodes with uncertain routes, \ie for $i \in \mathcal{N}$ with $f(i)=0$, and present a framework and methodology for \textit{probabilistic} inference of routes (Section~\ref{sec:probabilistic-inference}). We calculate the route probabilities $\pi_{i}(m)$ for all nodes $i\in\mathcal{N}$ and \ppoints $m\in\mathcal{M}$
.}

Next, we study how to enhance (certain or probabilistic) inference, when \textit{oracles} (\eg \textit{measurements}) are given for a set of nodes with uncertain routes (Section~\ref{sec:measurement-enhanced-inference}).

\blue{Finally, we consider the case where nodes \textit{prefer shorter paths} that conform to their routing policies (this frequently holds in practice~\cite{ciscobestpath,anwar2015investigating}), and incorporate this preference in our framework by modifying the \Rgraph; this enables route inference for more nodes (Section~\ref{sec:weak-path-inference}).}

\blue{The aforementioned inference methods (certain, probabilistic, with oracles) can be used independently or complementarily.} %
Table~\ref{table:methodology-overview} gives an overview of the inference methodology, namely, the sequence of steps (algorithms) needed for applying the different route inference variants proposed in this paper.

\begin{table}
\centering
\caption{Inference Methodology Overview.}
\label{table:methodology-overview}
\vspace{-2.5mm}
\begin{small}
\begin{tabular}{|cccc|l|}
\hline
\multicolumn{4}{|c|}{Type of Inference} & \multicolumn{1}{c|}{Methodology}\\
\hline
{\rotatebox[origin=c]{90}{\parbox[c]{1.8cm}{\centering Certain}}}	& {\rotatebox[origin=c]{90}{\parbox[c]{1.8cm}{\centering Probabilistic}}}	&	
{\rotatebox[origin=c]{90}{\parbox[c]{1.8cm}{\centering \blue{Oracles}
}}}	& {\rotatebox[origin=c]{90}{\parbox[c]{1.8cm}{\centering Shortest path preference}}}	&{\parbox[c]{5cm}{\centering Sequence of steps / algorithms. \\(*\textproc{Bel}: any exact or approximate \textit{belief updating} algorithm~\cite{van2008handbook})}}\\
\hline
{\yes}&{}&{}&{(\yes)}&{Alg.\ref{alg:build-vf-graph} $\Rightarrow$  (Alg.\ref{alg:transform-vf-bn-weak} $\Rightarrow$) Alg.\ref{alg:coloring-vf-graph}}\\
\hline
{}&{\yes}&{}&{(\yes)}&{Alg.\ref{alg:build-vf-graph} $\Rightarrow$ (Alg.\ref{alg:transform-vf-bn-weak} $\Rightarrow$) Alg.\ref{alg:coloring-vf-graph} $\Rightarrow$ Alg.\ref{alg:probabilistic-coloring-vf-graph}}\\
\hline
{\yes}&{}&{\yes}&{(\yes)}&{Alg.\ref{alg:build-vf-graph} $\Rightarrow$ (Alg.\ref{alg:transform-vf-bn-weak} $\Rightarrow$) Alg.\ref{alg:coloring-vf-graph} $\Rightarrow$ Alg.\ref{alg:probabilistic-coloring-vf-graph} $\Rightarrow$ Alg.\ref{alg:coloring-from-measurements}}\\
\hline
{}&{\yes}&{\yes}&{(\yes)}&{Alg.\ref{alg:build-vf-graph} $\Rightarrow$  (Alg.\ref{alg:transform-vf-bn-weak} $\Rightarrow$) Alg.\ref{alg:coloring-vf-graph} $\Rightarrow$ Alg.\ref{alg:probabilistic-coloring-vf-graph} $\Rightarrow$ \textproc{Bel}}\\
\hline
\end{tabular}
\end{small}
\end{table}


\myitem{Comparison to simulation models.} \blue{Simulation-based approach\-es
\cite{gao2001stable,c-bgp,muhlbauer2006building,feamster2004model} return a single outcome of catchment each time. Running a simulation more than once, may give different outcomes, since simulators typically employ randomization to determine the best path when not sufficient knowledge (\eg the $\mathcal{Q,H}$ or tie-breaker values) is available. For example, let the outcome for nodes $i$ and $j$ of the first (denoted in the superscript) simulation run be $\bestaspath{i}{n_{dst}}^{(1)}\rhd m1$ and $\bestaspath{j}{n_{dst}}^{(1)}\rhd m1$, and of the second run be $\bestaspath{i}{n_{dst}}^{(2)}\rhd m1$ and $\bestaspath{j}{n_{dst}}^{(2)}\rhd m2$. Based solely on these outcomes, one cannot answer the following questions: \textit{What would be the outcome of a third run? Will $i$ always route to $m1$, or did it happen twice due to random tie-breaking? Which route ($m1$ or $m2$) is more probable for $j$, if we simulate all possible tie-breaking combinations?}}

\blue{Our methodology provides answers to these questions (and with low complexity algorithms), where simulation-based models would need several simulation runs (of much higher complexity) to provide only an approximate answer. For instance, the certain inference algorithm (Section~\ref{sec:vf-graph-based-inference}) infers whether $i$ will always route to $m1$, and the probabilistic inference algorithm (Section~\ref{sec:probabilistic-inference}) calculates the percentage of all possible outcomes in which $j$ will route to $m1$.}

\subsection{Building the \Rgraph}\label{sec:build-vf-graph}
We design Algorithm~\ref{alg:build-vf-graph} to build the \Rgraph $\mathcal{G}_{R}$ that encodes all eligible paths to $n_{dst}$. Any eligible path $\aspath{i}{n_{dst}}$, $\forall i\in\mathcal{N}$, can be extracted by processing $\mathcal{G}_{R}$. Figure~\ref{fig:example-VF-graph} shows an example of a \Rgraph rooted in $n_{dst}$. 

\begin{figure}
\centering
\begin{minipage}{0.5\linewidth}
\includegraphics[width=1\linewidth]{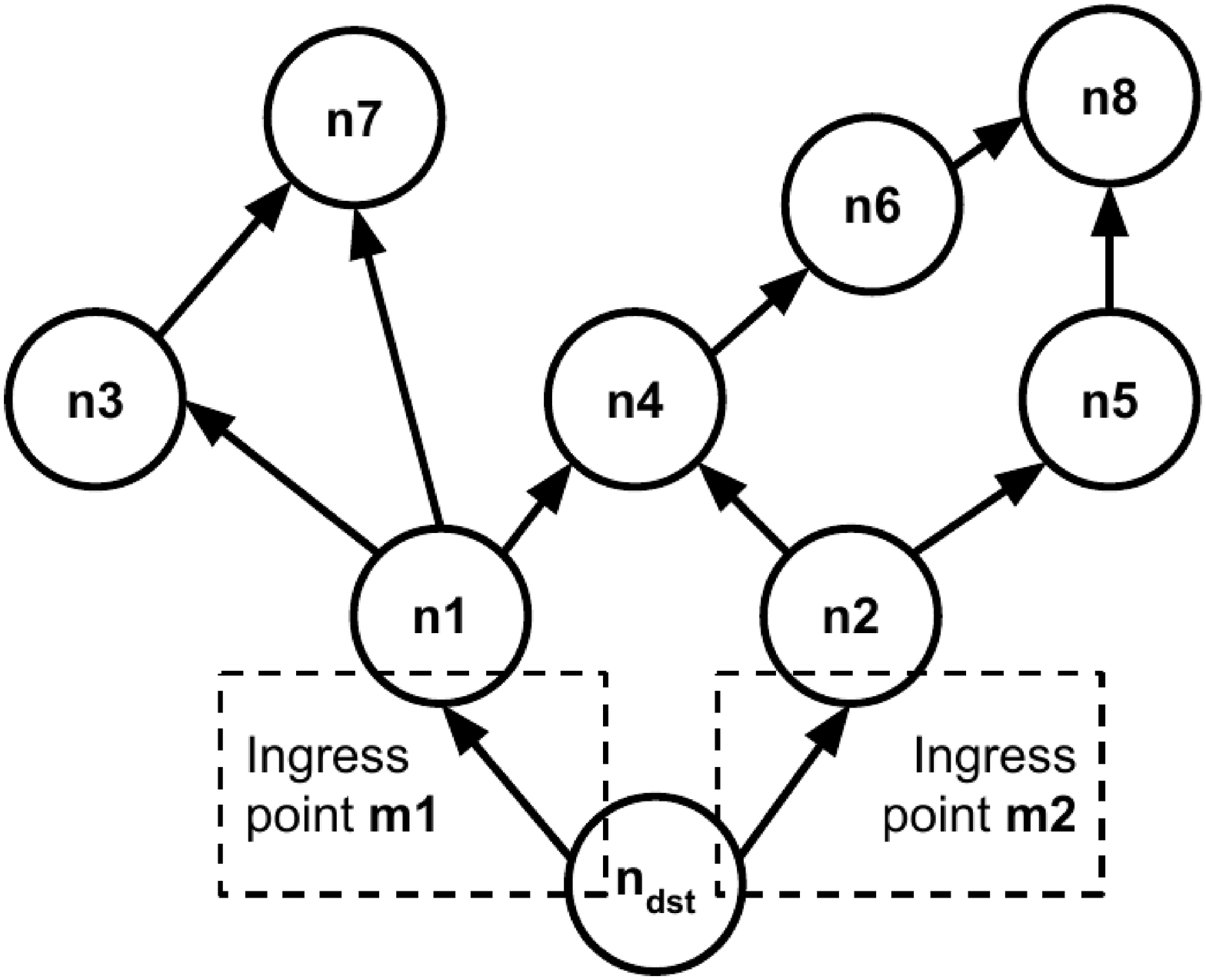}
\end{minipage}
\hspace{0.05\linewidth}
\begin{minipage}{0.4\linewidth}
\centering
\begin{small}
\begin{tabular}{|l|l|l|}
\hline
{\textbf{Node}} & {\textbf{Eligible path(s)}}&{$\mathbf{f(n)}$}\\
\hline
{$n1$} & {$[n1,n_{dst}]$}&{$m1$}\\
\hline
{$n2$} & {$[n2,n_{dst}]$}&{$m2$}\\
\hline
{$n3$} & {$[n3,n1,n_{dst}]$}&{$m1$}\\
\hline
{$n4$} & {$[n4,n1,n_{dst}]$}&{$0$}\\
{}     & {$[n4,n2,n_{dst}]$}&{} \\
\hline
{$n5$} & {$[n5,n2,n_{dst}]$}&{$m2$}\\
\hline
{$n6$} & {$[n6,n4,n1,n_{dst}]$}&{$0$}\\
{}     & {$[n6,n4,n2,n_{dst}]$}&{}\\
\hline
{$n7$} & {$[n7,n1,n_{dst}]$}&{$m1$}\\
{}     & {$[n7,n3,n1,n_{dst}]$}&{}\\
\hline
{$n8$} & {$[n8,n5,n2,n_{dst}]$}&{$0$}\\
{}     & {$[n8,n6,n4,n1,n_{dst}]$}&{} \\
{}     & {$[n8,n6,n4,n2,n_{dst}]$}&{}\\
\hline
\end{tabular}
\end{small}
\end{minipage}
\caption{Example of a \Rgraph (left), and the corresponding eligible paths and values of the routing function $f$ returned by Alg.~\ref{alg:coloring-vf-graph} for every node (right). The destination node $n_{dst}$ has two \ppoints $m1$ and $m2$ through which it connects to its neighbors $n1$ and $n2$, respectively.}
\label{fig:example-VF-graph}
\end{figure}


\myitem{Input/Output.} Algorithm~\ref{alg:build-vf-graph} receives as input a network graph and its routing policies $\mathcal{G}(\mathcal{N},\mathcal{E},\mathcal{Q},\mathcal{H})$, and a destination node $n_{dst} \in \mathcal{N}$. It returns as output the \Rgraph $\mathcal{G}_{R}(\mathcal{N}_{R},\mathcal{E}_{R})$, which is a DAG rooted in node $n_{dst}$.


\myitem{Workflow.} First, Algorithm~\ref{alg:build-vf-graph} simulates the operation of BGP; when needed, ``ties'' are broken randomly, \eg if multiple paths from neighbors with equal local preferences exist, one of them is selected randomly as the best path.
This randomness \emph{does not affect} the construction of the \Rgraph, since all incoming path advertisements exist in the RIB of a node $i$ ($\mathcal{P}_{i}$, returned in \textit{line 1}), and are taken into account (loop in \textit{line 6}).
Then, it initializes the \Rgraph by adding only the nodes, without adding any edge (\textit{line 2}). For each node $i$ (\textit{lines 3--18}), it accesses its RIB $\mathcal{P}_{i}$ and finds all neighbors that advertised a path for $n_{dst}$ (\ie the next-to-$i$ hops in the RIB paths; \textit{line 7}), and selects the set of the neighbors ($best\_neighbors$) with the highest local preference ($max\_q$). The paths from these neighbors are the eligible paths of $i$, since they (a) exist in the RIB and (b) are from neighbors with the highest local preference (as requested by Def.~\ref{def:eligible-path}). For each neighbor $k$ of $i$ in $best\_neighbors$, it adds a directed edge from $k$ to $i$. 

\myitem{Complexity: $O\left(|\mathcal{N}|\cdot |\mathcal{E}|\right)$.} \blue{The computational complexity of Algorithm~\ref{alg:build-vf-graph} is dominated by the complexity of running a BGP simulation (\textit{line 1}) which is equivalent to this of the centralized Bellman-Ford algorithm $O\left(|\mathcal{N}|\cdot |\mathcal{E}|\right)$. The loop in \textit{lines 3--18} examines every edge in the graph at most once and runs in $O(|\mathcal{E}|)$}.


\begin{algorithm}[t]
\caption{Building the \Rgraph. 
}
\label{alg:build-vf-graph}
\begin{small}
\begin{algorithmic}[1]
\Statex {\textbf{Input:} Network graph $\mathcal{G}(\mathcal{N},\mathcal{E},\mathcal{Q},\mathcal{H})$; destination node $n_{dst}$.}
\State $\mathcal{P}\leftarrow \textproc{RunBGP\_RandomTieBreak}(\mathcal{G}(\mathcal{N},\mathcal{E},\mathcal{Q},\mathcal{H}), n_{dst})$  
\Statex \LineComment{$\mathcal{P} = \{\mathcal{P}_{i}: i\in\mathcal{N}\}$, where $\mathcal{P}_{i}$ is the BGP RIB of $i$}
\State $\mathcal{G}_{R}(\mathcal{N}_{R},\mathcal{E}_{R}): \mathcal{N}_{R}\leftarrow\mathcal{N}; \mathcal{E}_{R}\leftarrow\emptyset$ 
\For {$i\in\mathcal{N}$}
	\State $best\_neighbors\leftarrow \emptyset$
    \State $max\_q \leftarrow -\infty$
	\For{$\aspath{i}{n_{dst}}\in \mathcal{P}_{i}$}
    	\State $k\leftarrow \textproc{GetNeighbor}(i,\aspath{i}{n_{dst}})$
    	\If{$q_{ik}< max\_q$}
        	\State \textbf{continue}
        \ElsIf{$q_{ik}> max\_q$}
        	\State $best\_neighbors\leftarrow \{k\}$
            \State $max\_q \leftarrow q_{ik}$
        \Else \LineComment{$q_{ik} = max\_q$}
        	\State $best\_neighbors\leftarrow best\_neighbors\cup\{k\}$
        \EndIf
    \EndFor
    \State $\mathcal{E}_{R}\leftarrow \mathcal{E}_{R}\cup \{e_{ki}:k\in best\_neighbors\}$
\EndFor
\State\Return $\mathcal{G}_{R}(\mathcal{N}_{R},\mathcal{E}_{R})$
\end{algorithmic}
\end{small}
\end{algorithm}

\blue{The following theorem formally states that (i) any path in the R-graph is eligible and (ii) any eligible path is encoded in the R-graph.}

\begin{theorem}\label{thm:vf-graph}
A path $\aspath{i}{n_{dst}}$ is an eligible path \emph{if and only if} it can be constructed by starting from $n_{dst}$ and following a sequence of directed edges in \Rgraph $\mathcal{G}_{R}$ until reaching $i$.
\end{theorem}
\begin{proof}
The proof is given in Appendix~\ref{appendix:proof-theorem-vf-graph}.
\end{proof}

\subsection{Route Inference on the R-Graph}\label{sec:vf-graph-based-inference}
We proceed to infer through which \ppoint a node $i$ routes its traffic to $n_{dst}$, by exploiting the structure of the \Rgraph. We demonstrate this inference using the example of Fig.~\ref{fig:example-VF-graph}, where it is given that $n1\rhd m1$ and $n2\rhd m2$ (\ie $f(n1)=m1$ and $f(n2)=m2$).

\myitemit{Case A: When the best path is known, the route inference is straightforward (from \eq{eq:routing-function}).} Node $n3$ has only one way/path to reach $n_{dst}$ (\ie by following links in the \Rgraph; see Theorem~\ref{thm:vf-graph}). This path is through node $n1$, and since $f(n1)=m1$, it follows that $f(n3)=f(n1)=m1$. 

\myitemit{Case B: Route inference is possible, even when the best path cannot be determined.} Node $n7$ has two incoming links from nodes $n1$ and $n3$; it selects only one of them to form its best path, based on its local preferences to $n1$ and $n3$. Without knowing these local preferences, we cannot infer the best path. However, since both $n1$ and $n3$ route traffic through the same \ppoint ($f(n1)=f(n3)=m1$), selecting either path leads to the same value of the routing function: $f(n7)=m1$. 

\myitemit{Case C: Route inference might not be possible for some nodes.} On the contrary to $n7$, while node $n4$ has also two incoming links, they are from nodes $n1$ and $n2$ for which it holds that $f(n1)\neq f(n2)$. Thus, in this case 
we cannot infer which path will be selected, and we write $f(n4)=0$.

The above rules can be applied sequentially for all nodes in the \Rgraph. Algorithm~\ref{alg:coloring-vf-graph} formalizes this inference process.

\myitem{Input/Output.} Algorithm~\ref{alg:coloring-vf-graph} receives as input a \Rgraph, a destination node $n_{dst}$ (root of the graph), and a mapping of the neighbors of $n_{dst}$ to its \ppoints $\mathcal{M}$. It returns the values of the routing function $f$ for all nodes in the \Rgraph.

\myitem{Workflow.} Algorithm~\ref{alg:coloring-vf-graph} starts from the neighbors of $n_{dst}$ and sets the values of $f$ according to their mapping to \ppoints (\textit{lines 2--7}). Then, it calculates a topological ordering%
\footnote{
A topological sort/ordering $\mathcal{T}$ of a directed graph $\mathcal{G}(\mathcal{N},\mathcal{E})$ is a linear ordering of its nodes $\mathcal{N}$ such that for every directed edge $e_{ij}\in\mathcal{E}$ from node $i\in\mathcal{N}$ to node $j\in\mathcal{N}$, $i$ comes before $j$ in the sort/ordering $\mathcal{T}$. For example, in Fig.~\ref{fig:example-VF-graph}, node numbering ($n1,..,n8$) corresponds to a topological ordering.}
of the \Rgraph nodes (\textit{line 8})
and sequentially visits nodes starting from those that are closer to the $n_{dst}$ (\textit{lines 9--20}). For each node $i$, it calculates the set of \textit{routes} (Def.~\ref{def:route-catchment}) of its parent nodes $CR_{i}$ (\textit{lines 10--14}), which are the candidate routes for node $i$. If some of the parents do not have a certain route ($0\in CR_{i}$) or there are more than one candidate routes ($|CR_{i}|\neq 1$), then it cannot make a certain route inference for node $i$, and sets $f(i)=0$  (\textit{lines 15--16}). Otherwise (\ie there is only one candidate route for $i$), an inference is made and the route of $i$ is set equal to this of its parent(s) (\textit{line 18}). 

\noindent\textit{Remark:} Visiting nodes in their topological order ensures correctness of the algorithm, \ie that the routing function of a node $i$ will not be mis-inferred (\eg $f(i)=0$ instead of $f(i)=m$, $m\in\mathcal{M}$). This is because all parent nodes of $i$, which are the only nodes that affect the route of this node, will have been visited before node $i$. 

\myitem{Complexity: $O(|\mathcal{N}_{R}|+ |\mathcal{E}_{R}|)$.} \blue{The topological sort in \textit{line 8} is of complexity $O(|\mathcal{N}_{R}|+ |\mathcal{E}_{R}|)$ and the loop in \textit{lines 9--20} is of complexity $O(|\mathcal{E}_{R}|)$ since it visits each edge in $\mathcal{E}_{R}$ exactly once.}


\begin{algorithm}[t]
\caption{Inference on the \Rgraph.
}
\label{alg:coloring-vf-graph}
\begin{small}
\begin{algorithmic}[1]
\Statex {\textbf{Input:} \Rgraph $\mathcal{G}_{R}(\mathcal{N}_{R}, \mathcal{E}_{R})$; destination node $n_{dst}$; mapping ($\rhd$) of the neighbors of $n_{dst}$ to \ppoints $\mathcal{M}$.}
\State $f(i)\leftarrow 0, ~~\forall i\in\mathcal{N}_{R}$ \LineComment{Initialization}
\State $dst\_neighbors\leftarrow\{i\in\mathcal{N}_{R}: e_{n_{dst},i}\in\mathcal{E}_{R}\}$
\For {$i\in dst\_neighbors,~m\in\mathcal{M}$}
  \If{$i\rhd m$}
  	\State $f(i)\leftarrow m$
  \EndIf
\EndFor
\State $\mathcal{T}\leftarrow$\textproc{TopologicalSort($\mathcal{G}_{R}$)}
\For {$i\in \mathcal{T}\backslash\{n_{dst}\cup dst\_neighbors\}$}
  	\State $CR_{i}\leftarrow \emptyset$ \LineComment{Candidate routes for node $i$}
    \State $P_{i}\leftarrow \{j\in\mathcal{N}_{R}: e_{ji}\in\mathcal{E}_{R}\}$ \LineComment{Parents of $i$}
	\For {$j\in P_{i}$}
    	\State $CR_{i}\leftarrow CR_{i}\cup \{f(j)\}$
    \EndFor
    \If{$(0\in CR_{i}) ~\textbf{or}~ (|CR_{i}|\neq 1)$}
    	\State $f(i)\leftarrow 0$
    \Else
    	\State $f(i)\leftarrow CR_{i}$
    \EndIf
\EndFor
\State \Return $f(i),~\forall i\in\mathcal{N}_{R}$
\end{algorithmic}
\end{small}
\end{algorithm}




\subsection{Probabilistic Route Inference}\label{sec:probabilistic-inference}

The goal of probabilistic inference is to calculate the route probabilities $\pi_{i}(m)$ 
(defined in \eq{eq:definition-route-probability}). Hence, even for nodes for which a certain inference is not possible
, the \textit{probabilities $\pi_{i}(m)$ can provide extra information that can be useful}, \eg to predict the total load per \ppoint by taking the expectation over the route probabilities:
\begin{equation}\label{eq:traffic-load}
    Traffic\_Load(m) = \textstyle\sum_{i\in\mathcal{N}} T_{i}\cdot \pi_{i}(m)
\end{equation}
where $T_{i}$ is the known traffic load from $i$ to $n_{dst}$ ($T_{i}$ can be estimated independently of the deployment/routing setup\blue{, \eg from Netflow statistics or} similarly to \blue{the system proposed in}~\cite{de2017broad}).  

\myitem{The \Rgraph as a Bayesian Network (BN).} To proceed to probabilistic route inference, we handle the \Rgraph as a Bayesian network (BN)%
\footnote{A BN is a directed acyclic graph (DAG), where a directed edge $e_{ij}$ denotes a dependence of node  $j$ on node $i$~\cite{van2008handbook}. We remind that the \Rgraph is a DAG that encodes routing path dependencies; \eg a directed edge $e_{ij}$ denotes that node $i$ is the next hop of $j$ in a path $\aspath{j}{n_{dst}}$ from $j$ to $n_{dst}$.}%
, where a node $i$ can take a value 
$m\in\mathcal{M}$, and the respective probability is given by $\pi_{i}(m)$. Based on BN properties (and the \textit{causality} in the \Rgraph\blue{, \ie children nodes select routes learned from their parents and not the opposite}), \blue{the following expression can be used to calculate the probabilities $\pi_{i}(m)$, from the probabilities of the parents ($P_{i}$) of $i$:} 
%
%
\begin{equation}\label{eq:route-probability-from-parents}
\pi_{i}(m) = \textstyle\sum_{j\in P_{i}} \pi_{j}(m)\cdot p_{ij}
\end{equation}
where $p_{ij}$ the probability for $i$ to prefer a path from $j$ than any other parent node, and $\sum_{j\in P_{i}}p_{ij}=1$.

Algorithm~\ref{alg:probabilistic-coloring-vf-graph} applies the above \blue{equation} 
and calculates the probabilistic route inference on a \Rgraph. 

\myitem{Input/Output.} Algorithm~\ref{alg:probabilistic-coloring-vf-graph} receives as input the \Rgraph, the \ppoints, the values of the routing function and the probabilities $p_{ij}$, and returns the route probabilities $\pi_{i}(m)$, $\forall i\in\mathcal{N},m\in\mathcal{M}$.

\myitem{Workflow.} Algorithm~\ref{alg:probabilistic-coloring-vf-graph} initializes all probabilities to zero (\textit{line 1}) and starts visiting all nodes according to a topological sort (\textit{lines 2--14}). If a visited node $i$ has a certain route $m$, then it sets the probability $\pi_{i}(m)$ equal to $1$ (\textit{lines 4--5}). Otherwise, it applies \eq{eq:route-probability-from-parents} 
to calculate $\pi_{i}(m)$ from the probabilities of the parent nodes (\textit{lines 7--13}). Visiting nodes in a topological order satisfies that the probability of all parent nodes $P_{i}$ will have been calculated before visiting $i$.

\myitem{Complexity: $O(|\mathcal{N}_{R}|+ |\mathcal{E}_{R}|)$.} \blue{Similarly to the certain inference methodology, the complexity of the topological sort in \textit{line 2} is $O(|\mathcal{N}_{R}|+ |\mathcal{E}_{R}|)$, and this of the loop in \textit{lines 3--14} is $O(|\mathcal{E}_{R}|)$. However, Algorithm~\ref{alg:probabilistic-coloring-vf-graph} is used with Algorithm~\ref{alg:coloring-vf-graph} (see Table~\ref{table:methodology-overview}), which means that the topological sort is already calculated in Algorithm~\ref{alg:coloring-vf-graph} and can be passed as input to Algorithm~\ref{alg:probabilistic-coloring-vf-graph}.}

\myitem{Setting the values of the probabilities $p_{ij}$.} %
\blue{
Algorithm~\ref{alg:probabilistic-coloring-vf-graph} and \eq{eq:route-probability-from-parents}, require the probabilities $p_{ij}$ to be known. We stress that these probabilities are not the local preferences 
$q_{ij}$ (which are equal for all the parents of a node in the \Rgraph; cf. Algorithm~\ref{alg:build-vf-graph}), but other criteria based on which a node will break ties, such as, the router IP address or the time of the received BGP announcements~\cite{wei2018does}. In some cases, these criteria (and the respective probabilities) can be inferred from past measurements,\eg~\cite{muhlbauer2006building}. However, given no prior knowledge on the criteria or in the case where the tie-breaker values change over time, the probabilities can be set to equal values (uniformly) for all parents in the \Rgraph, \ie $p_{ij}=\frac{1}{|P_{i}|},~\forall j\in P_{i}$.
}

\begin{algorithm}[t]
\caption{Probabilistic route inference on the \Rgraph.
}
\label{alg:probabilistic-coloring-vf-graph}
\begin{small}
\begin{algorithmic}[1]
\Statex {\textbf{Input:} \Rgraph $\mathcal{G}_{R}(\mathcal{N}_{R}, \mathcal{E}_{R})$; \ppoints $\mathcal{M}$; routing function $f(i)$, $\forall i\in\mathcal{N}_{R}$; probabilities $p_{ij}$, $\forall e_{ji}\in\mathcal{E}_{R}$.}
\State $\pi_{i}(m)\leftarrow 0,~~\forall i\in\mathcal{N}_{R},m\in\mathcal{M}$ \LineComment{Initialization}
\State $\mathcal{T}\leftarrow$\textproc{TopologicalSort($\mathcal{G}_{R}$)}
\For {$i\in \mathcal{T}$}
	\If{$f(i)\neq 0$}
    	\State $\pi_{i}(f(i))\leftarrow 1$        
    \Else
    	\State $P_{i}\leftarrow \{j\in\mathcal{N}_{R}: e_{ji}\in\mathcal{E}_{R}\}$ \LineComment{Parents of $i$}
    	\For{$j\in P_{i}$}
        	\For{$m\in\mathcal{M}$}
                \State $\pi_{i}(m)\leftarrow \pi_{i}(m) + \pi_{j}(m)\cdot p_{ij}$        
            \EndFor
        \EndFor
    \EndIf
\EndFor
\State \Return $\pi_{i}(m),~\forall i\in\mathcal{N}_{R},m\in\mathcal{M}$
\end{algorithmic}
\end{small}
\end{algorithm}



\subsection{ \blue{Inference under Oracles}}\label{sec:measurement-enhanced-inference}
\blue{We proceed to study how to enhance the certain or probabilistic route inference}, when an ``oracle'' for the value of the routing function for a set of nodes $\mathcal{X}$, $\mathcal{X}\subset\mathcal{N}$, with previously uncertain routes ($f(i)=0$, $\forall i\in\mathcal{X}$), is given. Obviously, the values of $f$ for nodes in $\mathcal{X}$ are trivially inferred (from the oracle). However, here we show that \textit{an oracle for the routing function for a set of nodes $\mathcal{X}$, enables route inference for a --potentially-- larger set of nodes $\mathcal{Y}$, $\mathcal{Y}\supseteq \mathcal{X}$}.


\myitem{``Oracles'' in reality.} In practice, an ``oracle''  can be obtained by a measurement, such as BGP messages/RIBs collected at some node, \eg through a route collector~\cite{routeviews} (passive measurement), or traceroutes/pings \camera{(see, \eg~\cite{de2017broad})} 
from a node towards the destination node $n_{dst}$ (active measurement). In the remainder, we consider oracles in the context of a measurement, however, our methodology is valid in the general case, independently of how the oracle is obtained.

\textit{Remark:} Actual measurements are applicable only in the case of an existing deployment\blue{, where a destination node $n_{dst}$ has already established connections and announces prefixes to its neighbors}. The measurement-enhanced inference can then be useful for lightweight route inference, \eg with only a few, instead of exhaustive~\cite{de2017broad}, measurements. \blue{However, the oracle-enhanced inference techniques can be useful for planning purposes (hypothetical scenarios) as well, \eg identifying the optimal locations for installing monitoring equipment to efficiently monitor future deployments and routing configurations (see, \eg Section~\ref{sec:measurements}).}

 
We use again
the example of Fig.~\ref{fig:example-VF-graph} to demonstrate the measure\-ment-enhanced inference metho\-dology. The basic inference methodology (Sections~\ref{sec:vf-graph-based-inference} and~\ref{sec:probabilistic-inference}) 
cannot infer with certainty the values $f$ for nodes $n4$, $n6$, and $n8$ (see right column of the table in Fig.~\ref{fig:example-VF-graph}). By conducting measurements for some of these nodes, the following cases of route inference are possible.

\myitemit{Case A: The routes of the measured nodes are directly inferred.} When we measure a node $i$, we either learn its best path (\eg from BGP data, traceroutes) or through which \ppoint $m$ it routes traffic to $n_{dst}$ (pings~\cite{de2017broad}). In both cases, we can directly infer $f(i)$.

\myitemit{Case B: The routes of the \textit{children} of measured nodes might be inferred.} If node $n4$ is measured, then the route of $n6$ can be directly determined as well, since the eligible paths for $n6$ are through $n4$, and thus it must hold $f(n6)=f(n4)$. However, if $n6$ is measured, it is not always possible to infer the route of $n8$ as well: if $f(n6)=m2=f(n5)$, then we can infer $f(n8)=m2$, whereas if $f(n6)=m1\neq f(n5)$, then we cannot infer with certainty the route of $n8$.

\myitemit{Case C: The routes of the \textit{parents} of measured nodes might be inferred.} If $n6$ is measured, then we can directly infer the route for $n4$ (since, as discussed above, it must hold $f(n6)=f(n4)$). If $n8$ is measured there are two cases: (i) if $f(n8)=m1$, then, since $f(n5)=m2$ (see Fig.~\ref{fig:example-VF-graph}), we can infer that $n8$ selects its best path through $n6$ and thus $f(n6)=f(n8)$; (ii) if $f(n8)=m2$, then we cannot infer with certainty through which node is the best path of $n8$, and, in contrast to the previous case, we cannot infer $f(n6)$. 

Algorithm~\ref{alg:coloring-from-measurements} is based on the aforementioned guidelines to enhance the route inference in a \Rgraph, given a set of oracles.

\myitem{Input/Output.} Algorithm~\ref{alg:coloring-from-measurements} receives as input a \Rgraph, the \ppoints, the values of the routing function $f$ and the probabilities $\pi$ (which are calculated by Algorithms~\ref{alg:coloring-vf-graph} and~\ref{alg:probabilistic-coloring-vf-graph}, respectively), and a set of oracles that map nodes to \ppoints. It returns the updated values of the routing function $f$.

\myitem{Workflow.} For each node $i\in\mathcal{X}$ for which an oracle is provided, Algorithm~\ref{alg:coloring-from-measurements} calls the function $\textproc{SetRoute}$, which updates the routing function $f$ and probabilities $\pi$ (\textit{lines 1--5}). Specifically, $\textproc{SetRoute}$ sets the value of the routing function equal to the one of the provided oracle (\textit{line 8}), and updates the probabilities for node $i$ (\textit{lines 9--10}). Then, it finds the subset $CP_{i}$ of the parent nodes $P_{i}$ of $i$, which may route (or actually route) through the same \ppoint with $i$ (\textit{lines 11--12}). These are the candidate nodes that can be in the best path $\bestaspath{i}{n_{dst}}$. If there is only one such candidate parent node ($|CP_{i}|=1$), then with certainty this node has the same route with $i$. Hence, in case the route for this node is not already inferred ($f(CP_{i})=0$), there is a new inference for this node and $\textproc{SetRoute}$ is called. After making the inferences for the parents of $i$ (\textit{lines 13--15}), the algorithm proceeds to inference for the children nodes of $i$  (\textit{lines 16--26}). For each child $j$ without an inferred route (\textit{line 16}), it collects the distinct values of the routing function of its parents $P_{j}$ (\textit{lines 18--22}). If there is only one such value $CR_{j}$, and $CR_{j}\neq 0$, then it means that all the parent nodes of $j$ route traffic to $CR_{j}$ (in fact, in this case it holds that $CR_{j}\equiv m$). Thus an inference for the route of $j$ is possible, and $\textproc{SetRoute}$ is called. Finally, $\textproc{SetRoute}$ returns the updated $f$ and $\pi$.

\myitem{Complexity: $O(|\mathcal{N_R}|)$.} \blue{The method \textproc{SetRoute} is called at most once per node (even if called recursively), \ie up to $|\mathcal{N_R}|$ times; for more details see Theorem~\ref{thm:efficiency-measurement-algo} and its proof in Appendix~\ref{appendix:thm-efficiency-measurement-algo}.}


\begin{algorithm}[t]
\caption{Enhancing inference with measurements.
}
\label{alg:coloring-from-measurements}
\begin{small}
\begin{algorithmic}[1]
\Statex {\textbf{Input:} \Rgraph $\mathcal{G}_{R}(\mathcal{N}_{R}, \mathcal{E}_{R})$; \ppoints $\mathcal{M}$; routing function $f$; probabilities $\pi$; set of measured nodes $\mathcal{X}$ and their mapping ($\rhd$) to \ppoints.}
\For {$i\in\mathcal{X},~m\in\mathcal{M}$}
	\If{$i\rhd m$}
        \State $(f,\pi)\leftarrow\textproc{SetRoute}(i,m,f,\pi,\mathcal{G}_{R})$
    \EndIf
\EndFor
\State \Return f

\Statex 
\Function{SetRoute}{$i,m,f,\pi,\mathcal{G}_{R}$}
\State $f(i)\leftarrow m$
\State $\pi_{i}(m)\leftarrow 1$
\State $\pi_{i}(d)\leftarrow 0,~~\forall d\neq m$
\State $P_{i}\leftarrow \{j\in\mathcal{N}_{R}: e_{ji}\in\mathcal{E}_{R}\}$ \LineComment{Parents of $i$}
\State $CP_{i}\leftarrow \{j\in P_{i}: \pi_{j}(m)>0\}$ \LineComment{Candidate parents}
\If{$|CP_{i}|=1$ \textbf{and} $f(CP_{i})=0$}
	\State $(f,\pi)\leftarrow\textproc{SetRoute}(CP_{i},m,f,\pi,\mathcal{G}_{R})$
\EndIf
\State $C_{i}\leftarrow \{j\in\mathcal{N}_{R}: e_{ij}\in\mathcal{E}_{R}, f(j)=0\}$ 
\LineComment{Children of $i$ without inferred route}
\For {$j\in C_{i}$}
	\State $P_{j}\leftarrow \{k\in\mathcal{N}_{R}: e_{kj}\in\mathcal{E}_{R}\}$ 
    \State $CR_{j}\leftarrow \emptyset$ \LineComment{Candidate routes for $j$}
    \For {$k\in P_{j}$}
    	\State $CR_{j}\leftarrow CR_{j}\cup \{f(k)\}$
    \EndFor
    \If{$|CR_{j}|= 1$ \textbf{and} $CR_{j}\neq 0$}
    	\State $(f,\pi)\leftarrow\textproc{SetRoute}(j,m,f,\pi,\mathcal{G}_{R})$
    \EndIf
\EndFor
\State \Return $(f,\pi)$
\EndFunction
\end{algorithmic}
\end{small}
\end{algorithm}

\myitem{Problem properties and complexity.} As discussed in Section~\ref{sec:probabilistic-inference}, the \Rgraph is a BN. When an oracle is given, the probabilities in this BN can be updated to infer extra routes. However, updating exactly the probabilities $\pi$ is NP-hard (Lemma~\ref{thm:update-prob-np-hard}), since the \Rgraph is a multiply-connected BN (and not a \textit{polytree})~\cite{cooper1990computational}. However, efficient algorithms to \textit{approximate} the updated probabilities $\pi$ exist~\cite{van2008handbook}.

\begin{lemma}\label{thm:update-prob-np-hard}
Updating the probabilities $\pi$ in the \Rgraph to their new values $\pi'$ when an oracle is given, is NP-hard.
\end{lemma}
\begin{proof} The proof is given in Appendix~\ref{appendix:thm-update-prob-np-hard}
.
\end{proof}

Algorithm~\ref{alg:coloring-from-measurements} is based on BN belief propagation methods~\cite{van2008handbook}. The main difference is that it does not aim to update exactly all the probabilities $\pi$, but only the probabilities whose new value $\pi'$ is either $1$ or $0$. This is sufficient for a certain route inference (for the nodes for which this is possible), and can take place in polynomial time, as Theorem~\ref{thm:efficiency-measurement-algo} states.   

\begin{theorem}\label{thm:efficiency-measurement-algo}
Algorithm~\ref{alg:coloring-from-measurements} updates the probabilities $\pi$ for all nodes $i$ for which $\max_{m} \pi_{i}'(m) = 1$ holds, in polynomial time \blue{$O(\mathcal{N}_{R})$}.
\end{theorem}
\begin{proof} 
The proof is given in Appendix~\ref{appendix:thm-efficiency-measurement-algo}.
\end{proof}

\subsection{Preference of Shorter Paths}\label{sec:weak-path-inference}
\blue{
The \Rgraph encodes all eligible paths, given the set of local preferences $\mathcal{Q}$. In practice, a node commonly \textit{prefers the shortest (in terms of AS-hops) among the paths learned from neighbors of equal local preference} (\ie its parents in the \Rgraph)~\cite{ciscobestpath}.
This common behavior is widely considered in related work as well, \eg~\cite{anwar2015investigating, gill2011let,c-bgp
}. Hence, route inference under the assumption of shortest path preference is relevant to real network operations.
}

\blue{
Here, we show how to incorporate the shortest path preference in our methodology. We do this in Algorithm~\ref{alg:transform-vf-bn-weak}, by modifying the \Rgraph to eliminate the eligible paths that are always longer and thus never preferred by a node.
}
Specifically, assuming preference of shorter paths, means that not all the paths in the \Rgraph are eligible anymore. For example, in the \Rgraph of Fig.~\ref{fig:example-VF-graph}, node $n7$ has two paths; however, the path through $n1$ is shorter and preferred. The path through $n3$ is not eligible anymore, and thus the edge between $n3$ and $n7$ must be removed.

\myitem{Input/Output.} Algorithm~\ref{alg:transform-vf-bn-weak} receives as input the \Rgraph, modifies it, and returns the modified \Rgraph.

\myitem{Workflow.} A minimum length (of eligible paths) $L_{i}$ is set for each node $i$, and is initialized to $0$ for $n_{dst}$, and to $\infty$ for every other node (\textit{line 1}). $L_{i}$ denotes the minimum length of the eligible paths $\aspath{i}{n_{dst}}$. A node will prefer the shorter paths, and thus the objective is to remove the longer paths of a node from the \Rgraph. To this end, starting from nodes closer to $n_{dst}$ and following a topological sort, the set of parents $P_{i}$ of the node $i$ is calculated, and the value of $L_{i}$ is set equal to the minimum value $L_{j}$, $j\in P_{i}$, plus one (\textit{lines 3--7}). The parents that have longer paths to $n_{dst}$ will never be preferred by a node $i$. Hence, the incoming edges to $i$ from such parents are removed from the \Rgraph (\textit{line 8}).  

\myitem{Complexity: $O(|\mathcal{N}_{R}|+ |\mathcal{E}_{R}|)$.} \blue{The complexity of the topological sort in \textit{line 2} is $O(|\mathcal{N}_{R}|+ |\mathcal{E}_{R}|)$, and this of the loop in \textit{lines 3--9} is $O(|\mathcal{E}_{R}|)$. Similarly, Algorithm~\ref{alg:transform-vf-bn-weak} is used with Algorithm~\ref{alg:coloring-vf-graph} (see Table~\ref{table:methodology-overview}), which means that the topological sort is calculated only once.}


\begin{theorem}\label{thm:weak-paths-increase-coloring}
Applying Algorithm~\ref{alg:transform-vf-bn-weak} on a \Rgraph, can only increase (not decrease) the set of nodes with certain routes. 
\end{theorem}
\begin{proof}
We provide a sketch of the proof in Appendix~\ref{appendix:thm-weak-paths-increase-coloring}.
\end{proof}

\begin{algorithm}[t]
\caption{\Rgraph transformation for shortest path preference.}
\label{alg:transform-vf-bn-weak}
\begin{small}
\begin{algorithmic}[1]
\Statex {\textbf{Input:} \Rgraph $\mathcal{G}_{R}(\mathcal{N}_{R}, \mathcal{E}_{R})$; destination node $n_{dst}$.}
\Statex {\textbf{Ouput:} (updated) \Rgraph $\mathcal{G}_{R}(\mathcal{N}_{R}, \mathcal{E}_{R})$.}
\State $L_{n_{dst}}\leftarrow 0$; $L_{i}\leftarrow \infty,~\forall i\in\mathcal{N}_{R}\backslash\{n_{dst}\}$ \LineComment{Initialization}
\State $\mathcal{T}\leftarrow$\textproc{topological\_sort($\mathcal{G}_{R}$)}
\For {$i\in \mathcal{T}\backslash\{n_{dst}\}$}
	\State $P_{i}\leftarrow \{j\in\mathcal{N}_{R}: e_{ji}\in\mathcal{E}_{R}\}$
	\For {$j\in P_{i}$}
    	\State $L_{i} \leftarrow \min\{L_{i}, L_{j}+1\}$
    \EndFor
    \State $\mathcal{E}_{R} \leftarrow \mathcal{E}_{R} - \{e_{ji}\in\mathcal{E}_{R}: L_{j}+1 > L_{i}\}$
\EndFor
\State \Return $\mathcal{G}_{R}(\mathcal{N}_{R}, \mathcal{E}_{R})$
\end{algorithmic}
\end{small}
\end{algorithm}

\section{Use Case: Efficient Measurements}\label{sec:measurements}
In this section, we investigate how to efficiently select measurements in order to increase the (certain) inference under a routing configuration. Specifically, we consider the following problem.  

\myitem{The problem.} \textit{Given a budget of $B$ measurements, what is the optimal set of nodes to be measured that maximizes the (certain) route inference 
in the \Rgraph?}

The above problem may emerge in the context of a number of measurement-related applications in the Internet, such as how to efficiently select a set of vantage points from which to trigger data-plane measurements (\eg select the best set of RIPE Atlas probes~\cite{ripeatlas}, given a limit on measurement credits), or how to optimally deploy monitoring infrastructure for passive (\eg route collectors) or active (\eg probes) measurements.

In the remainder, we study this problem: \blue{in Section~\ref{sec:max-inf-catch-problem}} we show that it is hard to be solved exactly or even approximated (since it requires exponential --to the number of nodes-- complexity), and \blue{in Section~\ref{sec:max-inf-catch-greedy} we} propose a greedy algorithm for efficient measurement selection, leveraging the \Rgraph{'s} structure and properties.

\subsection{Problem Formulation and Properties}\label{sec:max-inf-catch-problem}

\myitem{Problem formulation.} Let $\mathcal{X}$, $\mathcal{X}\subseteq \mathcal{N}_{R}$, be a set of nodes for which we have an oracle (i.e., route measurement), and let $x$, $x\in\mathcal{M}^{|\mathcal{X}|}$, the routes of nodes in $\mathcal{X}$ (\ie $x$ is a vector of size $|\mathcal{X}|$, taking values in state space $\mathcal{M}^{|\mathcal{X}|}$). 
We will denote $\mathcal{X}\rhd x$. For example, if $\mathcal{X}$ consists of three nodes $\{n1, n2, n3\}$, which route to \ppoints $\{m1,m2\}$ as follows: $n1\rhd m1,n2\rhd m1,n3\rhd m2 $, then we denote $x=\{m1,m1,m2\}$.

Given a set $\mathcal{X}$ and its routes $x$, we denote as $\mathcal{NC}_{R}(\mathcal{X}\rhd x)$ the \textit{number of nodes with a certain route} given these oracles:
\begin{equation}\label{eq:definition-conditional-coverage}
\mathcal{NC}_{R}(\mathcal{X}\rhd x) = |\left\{i\in\mathcal{N}_{R}: f(i)\neq 0 | \mathcal{X}\rhd x\right\}|    
\end{equation}
Note that 
we cannot know through which \ppoint each measured node routes its traffic before conducting a measurement. Hence, to evaluate the effectiveness of selecting a set of nodes, we consider all the possible measurement outcomes $x$, $x\in\mathcal{M}^{|\mathcal{X}|}$. To this end, we denote the \textit{expected number of nodes with a certain route}, under a set of measured nodes $\mathcal{X}$ as:
\begin{equation}\label{eq:objective-function-greedy}
E_{P}\left[\mathcal{NC}_{R}(\mathcal{X})\right] = \sum_{x\in\mathcal{M}^{|X|}} \mathcal{NC}_{R}(\mathcal{X}\rhd x)\cdot P(\mathcal{X}\rhd x)     
\end{equation}
where $P(\mathcal{X}\rhd x)$ denotes the probability of realization of the measurements outcome $x$.

Then, given a budget of at most $B$ measurements, and a set $\mathcal{Y}$, $\mathcal{Y}\subseteq\mathcal{N}_{R}$ of nodes which can be measured (e.g., for measurements with RIPE Atlas, $\mathcal{Y}$ can be the set of ASes that host at least one probe), the optimization problem can be expressed as\footnote{Generalizations of the problem can be expressed as well, \eg by weighting with $w_{i}$ (\eg based on the incoming traffic load from $i$) the importance of knowing the route of each node $i$ and modifying the definition of the  objective function in \eq{eq:definition-conditional-coverage} as $\mathcal{NC}_{R}(\mathcal{X}\rhd x) = \sum_{i\in\{\mathcal{N}_{R}: f(i)\neq 0 | \mathcal{X}\rhd x\}} w_{i}$, and/or assigning different measurement costs $c_{i}$ per node $i$ by modifying the constraint as $\sum_{i\in \mathcal{X}}c_{i}\leq B$.}:


\begin{problem}\label{problem:budgeted-measurements-certain}
~\hfill~$
\max_{\mathcal{X}\subseteq\mathcal{Y}} E_{P}\left[\mathcal{NC}_{R}(\mathcal{X})\right],~\text{~~~~}~s.t.~~|\mathcal{X}|\leq B
$
\end{problem}

\myitem{Modularity of the objective and the greedy algorithm.} Problem~\ref{problem:budgeted-measurements-certain} belongs to the class of combinatorial problems of maximizing a set function under a cardinality constraint. Lemma~\ref{lemma:properties-problem-certain} summarizes the properties of the objective function of Problem~\ref{problem:budgeted-measurements-certain}, which allow us to characterize its complexity and approximability. 
\begin{lemma}\label{lemma:properties-problem-certain}
The objective function of Problem~\ref{problem:budgeted-measurements-certain} is (i) non-negative and monotone, (ii) non-submodular, (iii) non-supermodular.
\end{lemma}
\begin{proof}
The proof is given in Appendix~\ref{appendix:thm-properties-modularity}.
\end{proof}

On the one hand, if the objective function of Problem~\ref{problem:budgeted-measurements-certain} was \textit{submodular}, then applying a greedy algorithm, of polynomial to $\mathcal{N}_{R}$ number of evaluations of the objective function $E_{P}\left[\mathcal{NC}_{R}(\mathcal{X})\right]$,  would come with an approximation guarantee of $1-1/e$ of the optimal solution~\cite{krause2012submodular}. On the other hand, if it was \textit{supermodular}, then the problem would be NP-hard to approximate~\cite{krause2012submodular}\footnote{Maximizing a super-modular function is equivalent to minimizing a sub-modular function, which is NP-hard when the size of the set is constrained.}. 
However, in the generic case of the \Rgraph we consider, with a monotone neither submodular, nor supermodular, objective, it has been recently shown that applying a greedy algorithm still comes with approximation guarantees (however, worse than in the case of a submodular function $1-1/e$) and usually the performance in practice is not far from the optimal~\cite{bian2017guarantees}. Therefore, in the following, we design a greedy algorithm for Problem~\ref{problem:budgeted-measurements-certain}, which starts with an empty set $\mathcal{X}^{0}=\emptyset$, and at each step $k$ adds to set $\mathcal{X}^{k-1}$ the node that increases the most the expected number of nodes with certain inference, i.e., 
\begin{equation*}
    \mathcal{X}^{k} = \mathcal{X}^{k-1}\cup \textstyle\arg\max_{i\in \mathcal{Y} \backslash \mathcal{X}^{k-1}} E_{P}\left[ \mathcal{NC}_{R}( \mathcal{X}^{k-1}\cup \{i\} ) \right]
\end{equation*}


\textit{Remark:} The approximation of the greedy algorithm depends on the \textit{submodularity ratio} and \textit{curvature} of the objective function, which in our case is determined by the structure of the \Rgraph~\cite{bian2017guarantees}. While deriving approximation guarantees as a function of structure and properties of the \Rgraph is an interesting research direction, it is out of the scope of this paper, and we defer it to future work. 

\myitem{Complexity in evaluating the objective.} A second challenge in solving Problem~\ref{problem:budgeted-measurements-certain}, even with a greedy algorithm, is that the evaluation of the objective function (\eq{eq:objective-function-greedy}) in each step, involves the calculation of the probabilities $P$, which may require also exponential to $|\mathcal{N}|$ time (see Section~\ref{sec:measurement-enhanced-inference}). We demonstrate this with the following example. Let $\mathcal{X}^{k}$ be the set of the first $k$ nodes selected by the greedy (or, any) algorithm, and a node $j\notin \mathcal{X}^{k}$. To evaluate the value of the objective function when adding node $j$ to the set of measurements, we need to proceed as follows:
\begin{align*}
E_{P}\left[ \mathcal{NC}_{R}( \mathcal{X}^{k}\cup \{j\} ) \right]
&= 
\sum_{x}\sum_{m} \mathcal{NC}_{R}( \mathcal{X}^{k} \cup \{j\} \rhd x\cup m)\cdot P(\mathcal{X}^{k}\cup\{j\}\rhd x\cup m) \\ 
&= \sum_{x}\sum_{m} \mathcal{NC}_{R}( \mathcal{X}^{k}\cup \{j\}\rhd x\cup m )\cdot P(j\rhd m | \mathcal{X}^{k}\rhd x) \cdot P(\mathcal{X}^{k}\rhd x)
\end{align*}
where we applied the Bayes theorem to express the joint probability as a product of the conditional probability. 


In the last equation, we can calculate the terms $\mathcal{NC}_{R}( \mathcal{X}^{k}\cup \{j\}\rhd x\cup m )$ using Algorithm~\ref{alg:coloring-from-measurements} (in $O(N)$ steps), and the terms $P(\mathcal{X}^{k}\rhd x)$ are already calculated in the $k-1$ step of the greedy algorithm. The remaining terms $P(j\rhd m | \mathcal{X}^{k}\rhd x)$ correspond to the updated probabilities $\pi_{j}$ for node $j$, given the set of oracles $\mathcal{X}^{k}\rhd x$. As discussed in Section~\ref{sec:measurement-enhanced-inference}, the exact calculation of the updated probabilities $\pi$ is NP-hard. 

In the greedy algorithm we propose, we trade accuracy for efficiency in the calculations for $\pi$ at each step, and update the probabilities $\pi$ with an approximate (``belief propagation'') method. 

\subsection{A Greedy Algorithm}\label{sec:max-inf-catch-greedy}
We present the greedy algorithm we propose for Problem~\ref{problem:budgeted-measurements-certain}, which is built upon the aforementioned guidelines. 

\myitem{Input/Output.} Algorithm~\ref{alg:greedy-measurements} receives as input a \Rgraph, the values of the routing function $f$ and the probabilities $\pi$, a set of nodes $\mathcal{Y}$ that are eligible to be measured, and a measurement budget $B$. It returns a set $\mathcal{X}$ of size $B$, containing the nodes to be measured
. 

\myitem{Workflow.} After the initialization (\textit{\red{line 1}}), Algorithm~\ref{alg:greedy-measurements} enters the greedy node selection loop (\textit{\red{lines 2--6}}), where at each iteration a node $i$ is added to the set of measured nodes $\mathcal{X}$ (\textit{\red{line 4}}). The node that is added is the one that --if measured-- increases the most the expected number of nodes with a certain route (\textit{\red{line 3}}). The expectation is calculated by \eq{eq:objective-function-greedy} using the probabilities $P$, i.e., 
\begin{align}
P(\mathcal{X}\cup\{j\}\rhd x\cup m) 
&= P(j\rhd m|\mathcal{X}\rhd x)\cdot P(\mathcal{X}\rhd x)
= \pi_{j}^{(\mathcal{X}\rhd x)}(m)\cdot P(\mathcal{X}\rhd x) \label{eq:conditional-probability-P-pi}
\end{align}
where we denote $\pi_{j}^{(\mathcal{X}\rhd x)}(m)=P(j\rhd m|\mathcal{X}\rhd x)$.~Note that $\pi_{j}^{(\emptyset\rhd x)}(m) = \pi_{j}(m)$. After adding node $i$ to set $\mathcal{X}$, the probabilities $\pi^{(\mathcal{X}\rhd x)}$ and $P(\mathcal{X}\rhd x)$, which will be needed in the next iteration, are calculated using the approximate method \textproc{UpdateProbabilities} (\textit{\red{line 5}}).


The method \textproc{UpdateProbabilities} calculates the probabilities $P(\mathcal{X}\cup\{j\}\rhd x\cup m)$ and $\pi_{j}^{(\mathcal{X}\rhd x)}(m)$, $\forall$ possible measurement outcomes (\textit{\red{lines 9--15}}). The former probabilities are calculated by using \eq{eq:conditional-probability-P-pi} and previous values (\textit{\red{line 10}}). The latter probabilities are calculated approximately in \textit{\red{lines 12--14}}. First, for the given outcome $\mathcal{X}\cup\{j\}\rhd x\cup m$, Algorithm~\ref{alg:coloring-from-measurements} is used to calculate the set of nodes with certain inference $\mathcal{Z}$ (\textit{\red{lines 12--13}}). We remind that Algorithm~\ref{alg:coloring-from-measurements} (called in \textit{\red{line 11}}) is a belief propagation method to update the probabilities of all the nodes with a certain route inference after a measurement/oracle is given. For the remaining nodes (with uncertain inference), we approximately update their probabilities by taking into account the inference for nodes in $\mathcal{Z}$ and applying only forward belief propagation in the \Rgraph (i.e., only in the direction of its edges). In other words, when a certain inference is made for a node $i\in\mathcal{Z}$, we consider only its effect on the probabilities of the (direct and indirect) children of $i$, and neglect the effect on the probabilities of its parents. This can be done by removing from the \Rgraph all the incoming edges to nodes in $\mathcal{Z}$ (\textit{\red{line 13}}), and then applying Algorithm~\ref{alg:probabilistic-coloring-vf-graph} (\textit{\red{line 14}}), which starts at the roots of the \Rgraph and through forward belief propagation calculates the probabilities $\pi$ for all nodes. 

Finally, we would like to remark that considering only forward belief propagation is the most reasonable choice in many use cases of our framework; namely, when the detailed preferences of a node $i$ to its parents $p_{ij}$ are not known and their values are arbitrarily set, \eg to equal values among all parents (see discussion in Section~\ref{sec:probabilistic-inference}).

\begin{algorithm}[t]
\caption{Selection of the set of nodes to be measured.
}
\label{alg:greedy-measurements}
\begin{small}
\begin{algorithmic}[1]
\Statex {\textbf{Input:} \Rgraph $\mathcal{G}_{R}(\mathcal{N}_{R}, \mathcal{E}_{R})$; \ppoints $\mathcal{M}$; routing function $f$; probabilities $\pi$; set of measurement-eligible nodes $\mathcal{Y}$, with $\mathcal{Y}\subseteq \mathcal{N}_{R}$; measurement budget $B$, with $B<|\mathcal{Y}|$.}
\State $\mathcal{X}\leftarrow \emptyset;~~P(\emptyset\rhd m)\leftarrow 1;~~P(i\rhd m) \leftarrow \pi_{i}(m),~~~\forall i\in\mathcal{Y}, m\in\mathcal{M}$
\LineComment{Initialization}
\While {$|\mathcal{X}| < B$} \LineComment{The greedy node selection loop}
    \State $i \leftarrow \arg\max_{j\in \mathcal{Y} \backslash \mathcal{X}} E_{(P,\pi)}\left[ \mathcal{NC}_{R}( \mathcal{X}\cup \{j\} ) \right]$
    \State $\mathcal{X} \leftarrow \mathcal{X}\cup \{i\}$
    \State $\pi,P \leftarrow \textproc{UpdateProbabilities}\left(\mathcal{X}, i, \pi , P\right)$
\EndWhile
\State \Return $\mathcal{X}$

\Statex 
\Function{UpdateProbabilities}{$\mathcal{X}, i, \pi, P $}
\For{$x\in \mathcal{M}^{|\mathcal{X}|}, m\in\mathcal{M}$} 

\State $P(\mathcal{X}\cup\{j\}\rhd x\cup m) \leftarrow P(\mathcal{X}\rhd x)\cdot \pi_{i}^{(\mathcal{X}\rhd x)}(m)$
\State $f\leftarrow \textproc{Algorithm\ref{alg:coloring-from-measurements}}(\mathcal{G}_{R},\mathcal{M},f,\pi^{(\mathcal{X}\rhd x)},\mathcal{X}\cup \{i\}\rhd x\cup m) $
\State $\mathcal{Z} \leftarrow \left\{j\in\mathcal{N}_{R}: f(j)\neq 0\right\}$
\State $\mathcal{G}^{'}_{R}\leftarrow \mathcal{G}_{R}\left(\mathcal{N}_{R}, \mathcal{E}_{R} - \{e_{j\ell}\in\mathcal{E}_{R}:\ell\in \mathcal{Z}\}\right)$
\State $\pi^{(\mathcal{X}\rhd x,i\rhd m)} \leftarrow \textproc{Algorithm\ref{alg:probabilistic-coloring-vf-graph}}( \mathcal{G}^{'}_{R},\mathcal{M},f )$

\EndFor
\State \Return $\pi, P$
\EndFunction
\end{algorithmic}
\end{small}
\end{algorithm}

\section{Performance Evaluation}\label{sec:results}
In this section, we apply the proposed methods to the Internet AS-graph. Using realistic simulations, we evaluate the capability of our methodology to infer Internet routes, and discuss related insights. 

\subsection{Setup}\label{sec:results-sim-setup}
We build the AS-level topology using the experimentally collected CAIDA AS-relationship dataset~\cite{caidaasrel}. This contains a list of 
$\sim452k$ peering links between 
$\sim62k$ ASes, and their relationships $\ell_{ij}\in\{p2c,p2p,c2p\}$. We \blue{consider a single node per AS, valley-free routing, and set the policies according to Section~\ref{sec:vf-model} (\eq{eq:vf-local-preferences} and \eq{eq:vf-export-policies}). In the simulations, we break ties for routes received from neighbors of the same type (\eg from two customers) arbitrarily. However, in the inference, we assume that we do not know how exactly the nodes break ties (otherwise inference would be trivial). To account for this (assumed) lack of knowledge, we consider in the inference the more generic values $\hat{q}_{ij}=\hat{q}_{ik}\Leftrightarrow\ell_{ij}=\ell_{ik}$, \ie equal preferences for all neighbors of the same type. This takes into account all possible tie-breaking outcomes, and corresponds to a practical scenario, where we would like to infer catchment with coarse knowledge of the policies ($\ell_{ij}$).
} 


At each simulation, we create a new node $n_{dst}$, add it to the topology, and add \textit{c2p} links (with $n_{dst}$ the customer) to $|\mathcal{M}|$ randomly selected nodes; these $|\mathcal{M}|$ nodes are assumed to be connected in different \ppoints of $n_{dst}$. 
We announce a prefix from $n_{dst}$, and run (simulate) BGP
. For each different scenario setup, we conduct $1000$ simulation runs.

\subsection{Gains from the \Rgraph-based Inference}\label{sec:results-sim-gains}
A main contribution of the proposed methodology (basic/certain inference, Sections~\ref{sec:build-vf-graph} and~\ref{sec:vf-graph-based-inference}
) is that it achieves to (i) encode all eligible paths in a simple graph, and (ii) exploit the structure of the \Rgraph to infer routes even for nodes with multiple eligible paths. The simulation results in Fig.~\ref{fig:gains-basic-methodology} quantify these~gains.

\begin{figure}
\centering
\subfigure[Proposed vs.~Naive Inference]{\includegraphics[width=0.45\linewidth]{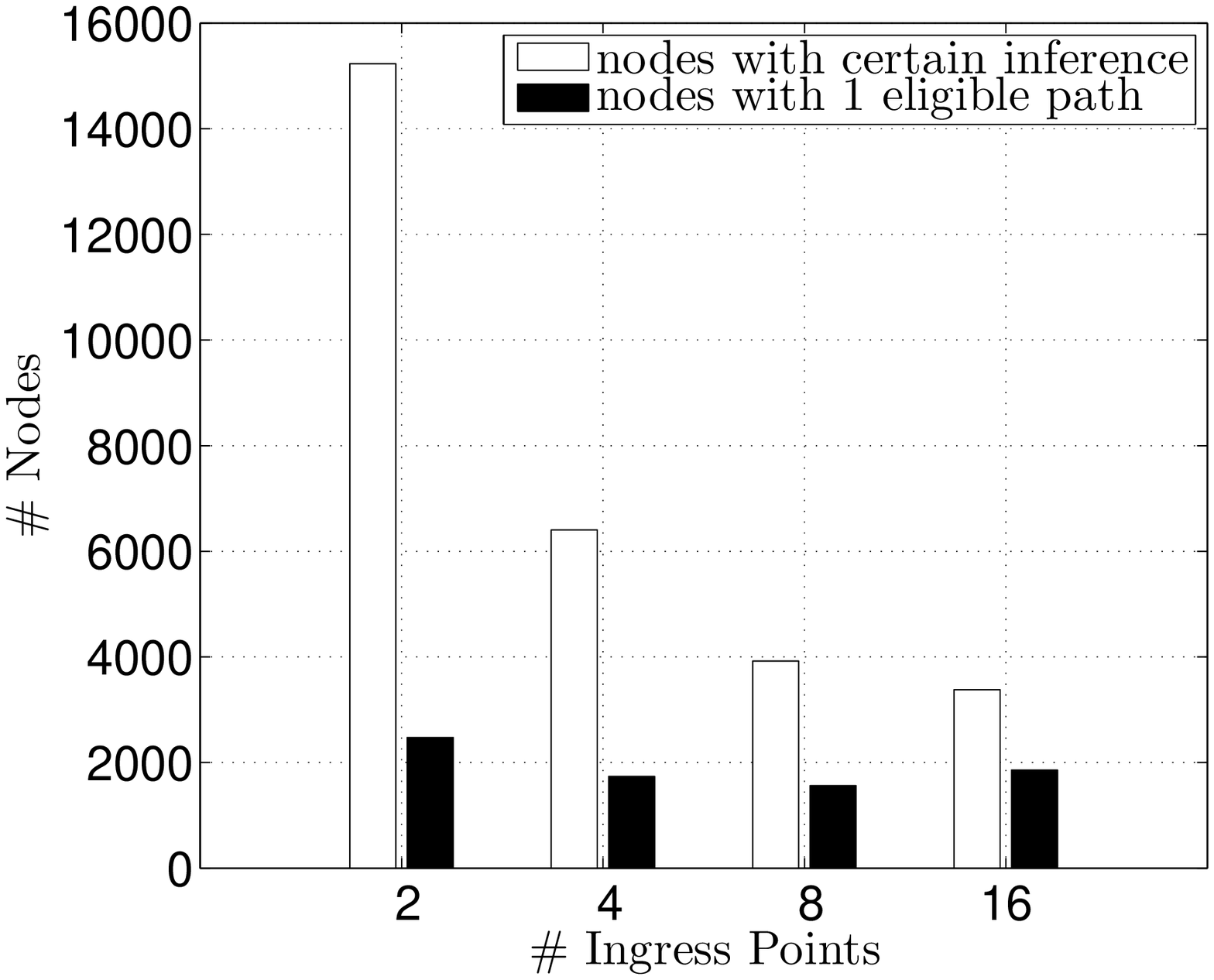}\label{fig:avg-nb-inference}}
\hspace{0.05\linewidth}
\subfigure[Eligible Paths]{\includegraphics[width=0.45\linewidth]{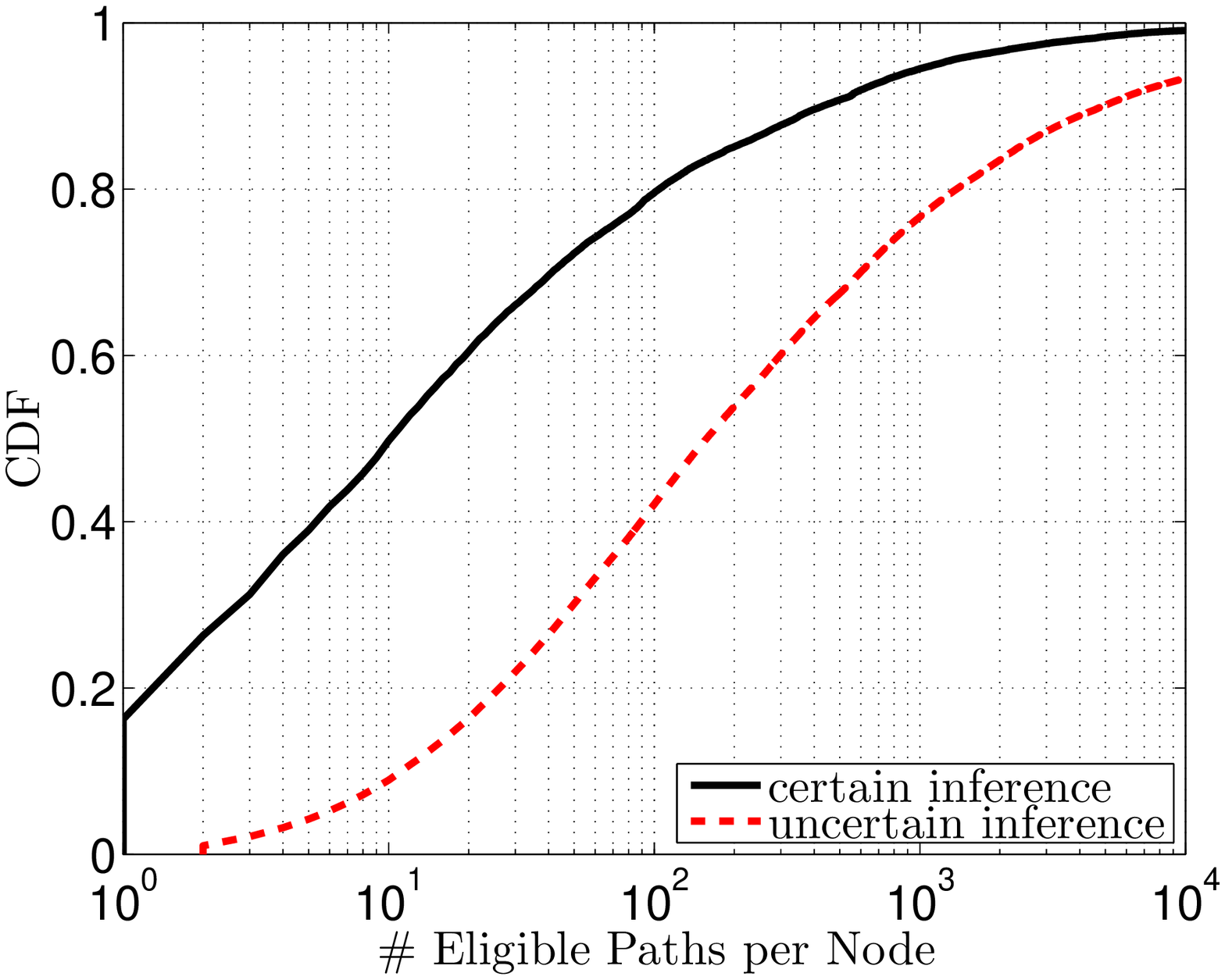}\label{fig:cdf-eligible-paths}}
\vspace{-5mm}
\caption{(a) Number of nodes with certain inference, and number of nodes with one eligible path. (b) Distribution of the number of eligible paths per~node, for nodes with certain and uncertain inference; setup with 2 \ppoints.}
\label{fig:gains-basic-methodology}
\end{figure}

Figure~\ref{fig:avg-nb-inference} compares the average number of nodes for which our methodology inferred a certain route (white bars), and the average number of nodes with only one eligible path (black bars). For scenarios in which the network has two \ppoints (leftmost bars), \textit{our methodology infers the routes of almost an order of magnitude more nodes than a naive inference} (that infers routes only for nodes with a single eligible path). As the number of \ppoints increases, the number of eligible paths --and thus the uncertainty-- increases as well; however, even for a large number of \ppoints (rightmost bars), our methodology infers around two times more nodes than a naive approach.

Moreover, Fig.~\ref{fig:cdf-eligible-paths} shows that $50\%$ (0.5 in y-axis) of the nodes for which an inference can be made (continuous line), have more than $10$ eligible paths (x-axis); respectively, in $20\%$ of the inferences ($0.8$ in y-axis) the nodes have more than $100$ eligible paths. This further highlights the gains from exploiting the structure of the \Rgraph towards making certain inferences.

\subsection{\Rgraph vs. Simulation-based Inference}\label{sec:rgraph-vs-sims}
\blue{As discussed earlier, one could use simulation-based approaches to estimate the catchment~\cite{gao2001stable,c-bgp,muhlbauer2006building}. Since each simulation run returns a single outcome (which is affected by the randomness in tie-breaking), several runs are needed to calculate estimates of the catchment. On the contrary, our methodology \textit{exactly} calculates the statistics for catchment, in a \textit{lightweight} way (computational complexity is approximately equal to one simulation run). The results in Fig.~\ref{fig:rgraph-vs-sims} demonstrate these advantages of our methodology.}

\blue{In Fig.~\ref{fig:rg-vs-sims-space}, we present results for \red{5} indicative scenarios (x-axis); in each of them $n_{dst}$ is connected to a randomly selected pair of \ppoints, \ie $\mathcal{M}=\{m1,m2\}$. For each scenario we do the following. (i) We run 1000 simulations, assuming shortest path preference, and for each run we measure the catchment of $m1$; we present the distribution of the results in boxplots (\texttt{SIMS}). (ii) We then apply our methodology to calculate the following quantities: \\
\textit{\textbf{Lower (LOW) and Upper (UPP) bounds}}: The certain catchment $|CC(m1)|$ of $m1$, where $CC(m1) = \{i\in\mathcal{N}: f(i)=m1\}$, (calculated by Algorithm~\ref{alg:coloring-vf-graph}) is a \textit{lower bound} for the catchment of $m1$, since more nodes (whose inference is not certain) may route to $m1$ as well. Respectively, an \textit{upper bound} for the catchment of $m1$ is given by $|\mathcal{N}|-|CC(m2)|$, since the nodes in $CC(m2)$ cannot route to $m1$. We present the lower/upper bounds with (\texttt{RG-LOW-SP} / \texttt{RG-UPP-SP}) and without (\texttt{RG-LOW-NO-SP} / \texttt{RG-UPP-NO-SP}) shortest path preference.\\
\textit{\textbf{Mean value (AVG)}}: We calculate the mean value of the catchment for $m1$ as $\sum_{i}\pi_{i}(m1)$, where the route probabilities $\pi_{i}$ are calculated by the probabilistic inference Algorithm~\ref{alg:probabilistic-coloring-vf-graph}, with (\texttt{RG-AVG-SP}) and without (\texttt{RG-AVG-NO-SP}) assuming shortest path preference.}

\begin{figure}
\centering
\subfigure[R-graph vs simulations in space]{\includegraphics[width=0.45\linewidth]{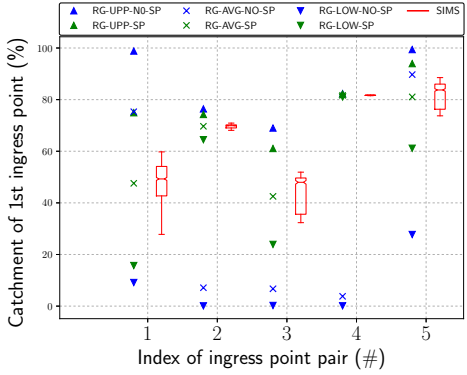}\label{fig:rg-vs-sims-space}}
\hspace{0.05\linewidth}
\subfigure[R-graph vs simulations in time]{\includegraphics[width=0.45\linewidth]{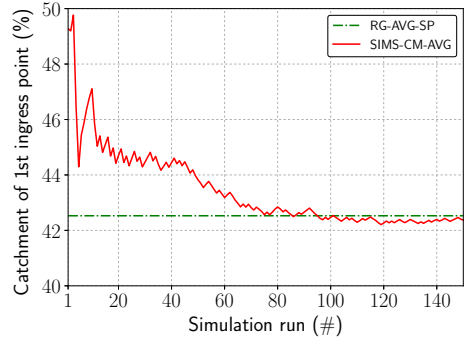}\label{fig:rg-vs-sims-time}}
\vspace{-5mm}
\caption{(a) Simulation results and inferences for the catchment of 1st \ppoint ($m1$) for different scenarios. (b) Catchment of 1st \ppoint ($m1$) for a scenario, calculated as a cumulative moving average of different simulation runs (\texttt{CM-AVG}) and as predicted by our methodology (\texttt{RG-AVG-SP}).
}
\label{fig:rgraph-vs-sims}
\end{figure}

\blue{
Some main observations and insights from the comparison of the simulation results (\texttt{SIMS}) with our predictions are: \\
(i) The predicted mean values \texttt{RG-AVG-SP} with shortest path preference (\ie as in the simulation setup) \textit{coincide always} with the average values calculated from the simulation results (\texttt{SIMS}). Note though that our prediction requires only a single simulation run, whereas simulation-based approaches require several runs to converge to the mean value. We demonstrate this in Fig.~\ref{fig:rg-vs-sims-time}, which shows that the average value of catchment calculated from simulation runs (continuous line), \textit{needs almost 100 runs to converge} to the predicted mean value (dashed line). The presented results are for the \red{3rd} scenario of Fig.~\ref{fig:rg-vs-sims-space}; however, similar patterns were observed across all the pairs we examined.
\\
(ii) As expected, none of the simulation results is outside the bounds \texttt{RG-LOW-SP} and \texttt{RG-UPP-SP}. When the upper and lower bounds are closer, simulation results are more concentrated around the mean. \textit{The distance between the lower/upper bounds shed light on the effect of the randomness in a simulation}. For example, in the \red{4th} scenario, the bounds coincide, thus showing that the catchment is not affected by the tie-breaking process; or, equivalently, knowing only coarse estimates of the policies is enough for an accurate prediction. On the other hand, in the \red{1st} scenario, the distance between the bounds is larger, which implies that measurements would be needed for an accurate calculation of the catchment. \\
(iii) The bounds that are calculated without assuming shortest path preference (\texttt{NO-SP}) are looser, since they account for a larger set of possible scenarios. The difference between the predictions with (\texttt{SP}) and without (\texttt{NO-SP}) shortest path preference reveals the effect of the path lengths in a routing configuration. This knowledge can be useful in traffic engineering. For example, in the \red{2nd} scenario, the two \textit{upper} bounds are very close. This means that the maximum catchment of $m1$ is affected mainly by the local preferences and not by the path lengths. Hence, even if we increase the length of the paths to $m2$ through path prepending, this would not increase significantly the catchment of $m1$. On the contrary, the large distance between the two \textit{lower} bounds in the same \red{2nd} scenario, indicates that applying path prepending to announcements through $m1$, can significantly decrease its catchment.}

\subsection{Completeness of Inference}\label{sec:results-sim-completeness}

In Fig.~\ref{fig:avg-nb-inference} we see that, \eg for $|\mathcal{M}|=2$ a certain inference is possible for $\sim15k$ of the total $\sim62k$ nodes in the graph. Here, we investigate for how many nodes (completeness) our methodology returns a \textit{certain} inference, with or without measurements. \textit{Remark:} \textit{probabilistic} inference is made for all nodes (see Section~\ref{sec:probabilistic-inference}).

To consider realistic scenarios, we simulate measurements from the vantage points of several real Internet measurement platforms:
\begin{itemize}[leftmargin=*]
\item \textbf{RouteViews}~\cite{routeviews} and \textbf{RIPE RIS}~\cite{riperis}  (\texttt{RV\_RIS}) provide BGP RIBs and updates collected from more than 400 ASes  worldwide. 
\item \textbf{RIPE Atlas}~\cite{ripeatlas} comprises more than $25k$ probes (in $\sim3.5k$ ASes), \ie devices able to run pings (\texttt{RA\_PING}) or traceroutes (\texttt{RA\_TRACE}) towards certain Internet destinations.
\item \textbf{Looking Glasses} (\texttt{LG}) are servers that provide the BGP RIBs of  the networks (ASes) they are hosted in. We use the Periscope platform~\cite{caidaperiscope}, to obtain a list of LGs in 883 ASes.
\end{itemize}

\noindent\textit{Remark:} The BGP data of a network $i$ or traceroutes from $i$ to $n_{dst}$ can provide a route oracle for all the nodes in the best path $\bestaspath{i}{n_{dst}}$. Pings from $n_{dst}$ to $i$ can provide a route oracle only for $i$~\cite{de2017broad}.

\begin{figure}
\centering
\subfigure[Base Model]{\includegraphics[width=0.49\linewidth]{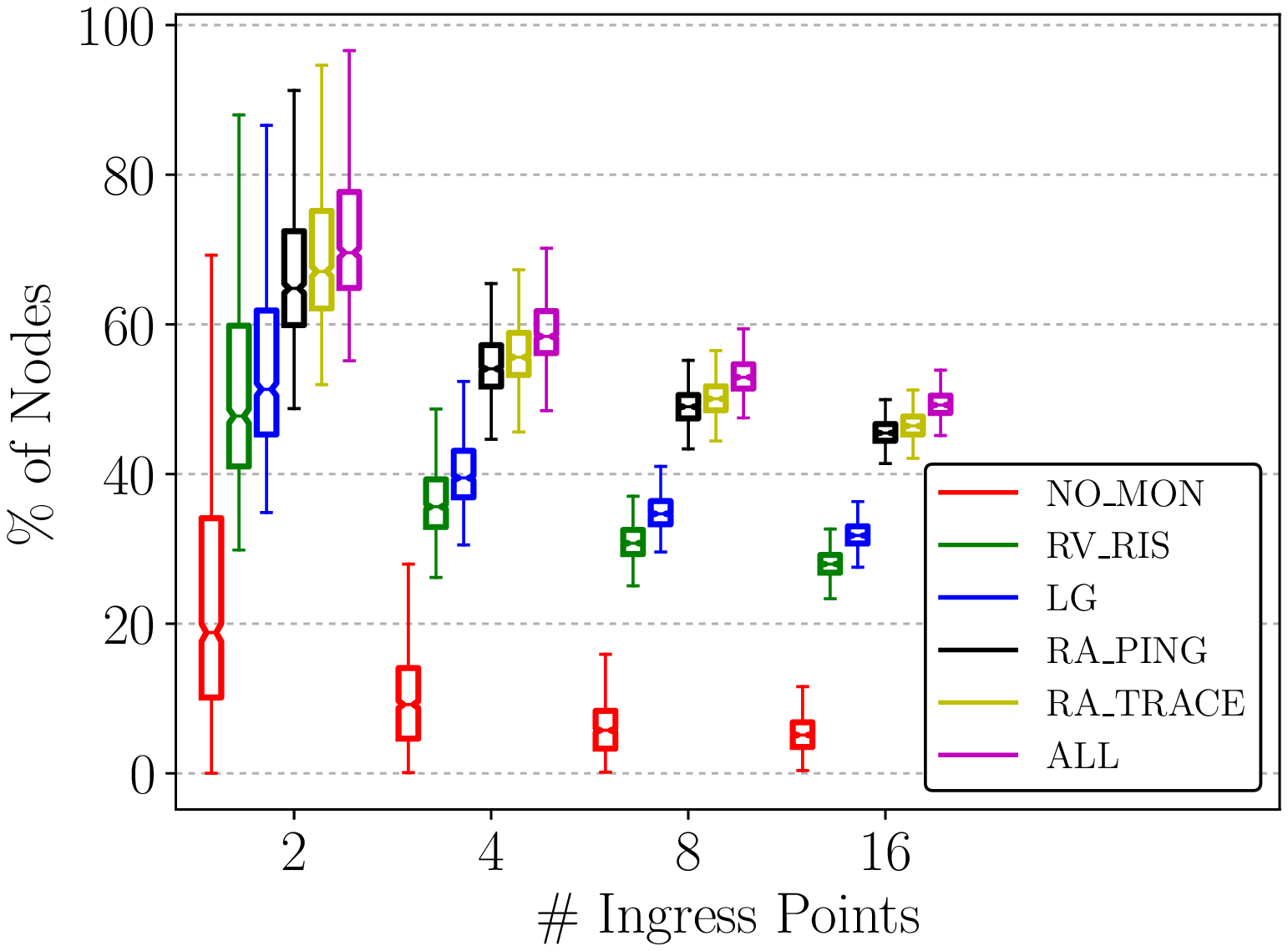}\label{fig:color_cov_perc_strict}}
\subfigure[Base Model + SP Preference]{\includegraphics[width=0.49\linewidth]{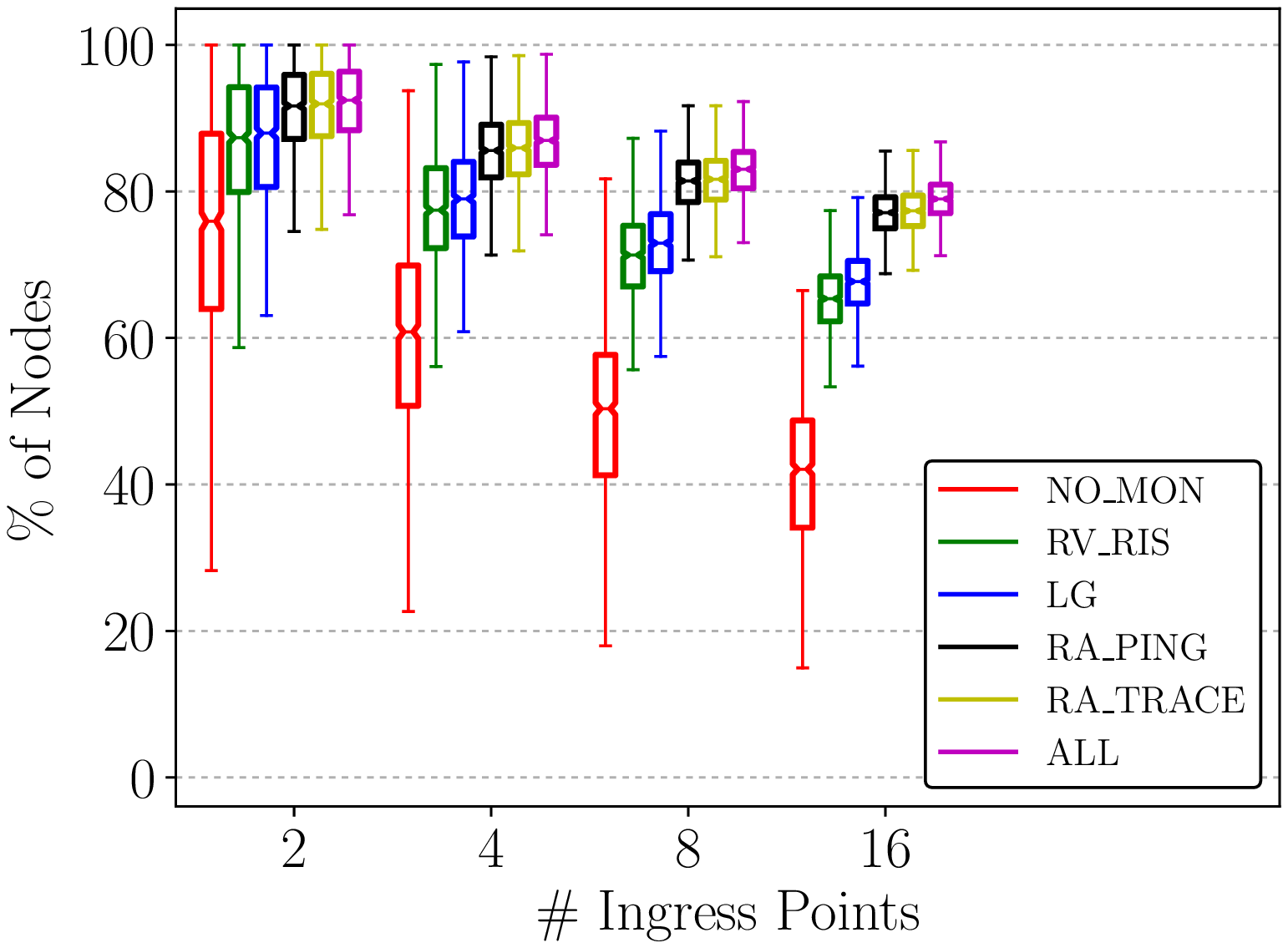}\label{fig:color_cov_perc_weak}}
\vspace{-5mm}
\caption{Distribution (over $1000$ simulation runs) of the percentage of nodes with a certain inference (y-axis), in scenarios with different number of \ppoints (x-axis). Boxplots correspond to different setups without and with measurements, for scenarios (a) \textit{without} and (b) \textit{with} preference~of~shorter~paths (SP).}
\label{fig:color_cov_perc}
\end{figure}

Figure~\ref{fig:color_cov_perc} shows for how many nodes a certain inference is possible in different setups. Some key observations are:

\noindent(i) The number of inferences decreases with the number of \ppoints. Although this is expected, our methodology quantifies how this behavior is affected by different parameters (number of \ppoints, measurement setup, etc.). 

\noindent(ii) Assuming preference of shorter paths leads to significantly more inferences. For instance, even without measurements (red boxplots - \texttt{NO\_MON}) the median percentages increase 
from $5\%-19\%$ (Fig.~\ref{fig:color_cov_perc_strict}) to $42\%-76\%$ (Fig.~\ref{fig:color_cov_perc_weak}). 

\noindent(iii) Public measurement platforms can significantly enhance inference. Their contribution is crucial when many \ppoints are in use; \eg in Fig.~\ref{fig:color_cov_perc_strict} for $|\mathcal{M}|=16$ using all platforms increases inference from $5\%$ to $49\%$, and in Fig.~\ref{fig:color_cov_perc_weak} from $42\%$ to $79\%$. 

\noindent(iv) Interestingly enough, even a lightweight measurement campaign with pings (black boxplots - \texttt{RA\_PING}; \eg as suggested in~\cite{de2017broad}), can achieve almost the same enhancement with employing all platforms together. However, we simulated pings only to RIPE Atlas probes ($3.5k$ measurements), in contrast to~\cite{de2017broad} that requires orders of magnitude more measurements; combining our methodology with that technique could potentially lead to even more efficient route inference.

\subsection{Efficient Measurements}
\blue{
Next, we evaluate the ability of Algorithm~\ref{alg:greedy-measurements} to select a set of nodes to be measured.
We consider two scenarios with $|\mathcal{M}|=2$, taking into account only the $\sim20k$ non-stub nodes of the AS-graph, and apply Algorithm~\ref{alg:coloring-vf-graph} to calculate the certain inference. The number of nodes whose routes cannot be inferred with certainty are 2975 (``low'') and 11918 (``high'') in the scenarios of Fig.~\ref{fig:greedy-low} and Fig.~\ref{fig:greedy-high}, respectively. To enhance inference, we conduct measurements to a set of nodes $\mathcal{X}$, and then apply Algorithm~\ref{alg:coloring-from-measurements}. 
}
\blue{
In Fig.~\ref{fig:greedy} we present results for the extra number of nodes whose routes are inferred with certainty after the measurements. Sets $\mathcal{X}$ are selected with the greedy Algorithm~\ref{alg:greedy-measurements} among 100 and 1000 nodes with RIPE Atlas probes (continuous lines, ``Greedy''), or are selected randomly among nodes with RIPE Atlas probes (dashed line, ``Random''). The main observation is that selecting the nodes to be measured with the proposed algorithm is significantly more efficient than a random selection. Our algorithm is able to select a good set of nodes, and its efficiency increases when the set of available nodes (with probes) is larger (``Greedy 1000'' vs. ``Greedy 100''). Comparing the results in Fig.~\ref{fig:greedy-low} and~\ref{fig:greedy-high}, reveals that the careful selection of the set $\mathcal{X}$ is more crucial for scenarios with high --initial-- uncertainty (Fig.~\ref{fig:greedy-high}); \eg a single measurement from the node selected with our algorithm can infer with certainty up to 1000 extra routes (``Greedy 1000'' in Fig.~\ref{fig:greedy-high}).}

\begin{figure}
\centering
\subfigure[2975 uncertain nodes]{\includegraphics[width=0.4\linewidth, height=0.35\linewidth]{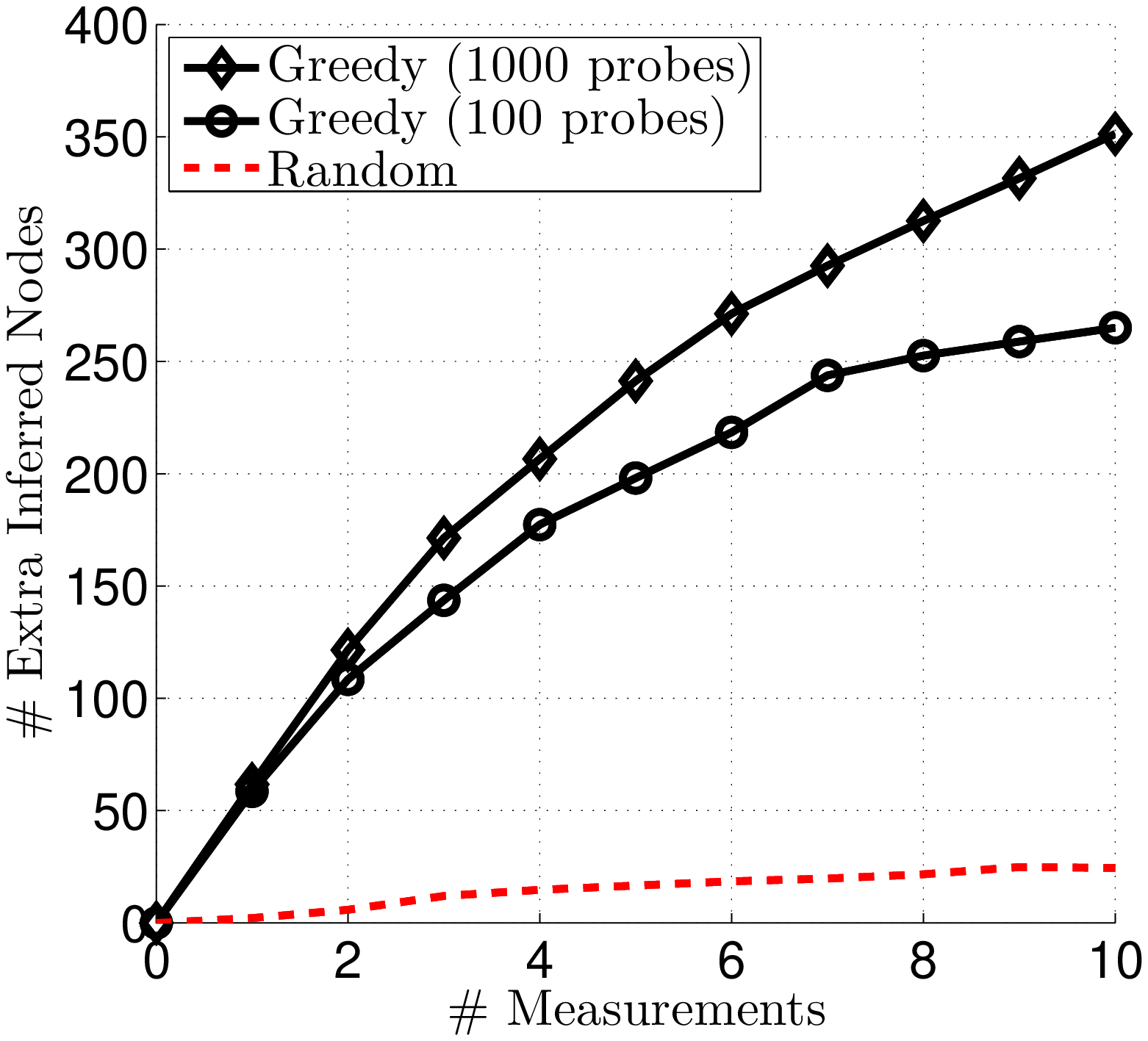}\label{fig:greedy-low}}
\hspace{0.05\linewidth}
\subfigure[11918 uncertain nodes]{\includegraphics[width=0.4\linewidth, height=0.35\linewidth]{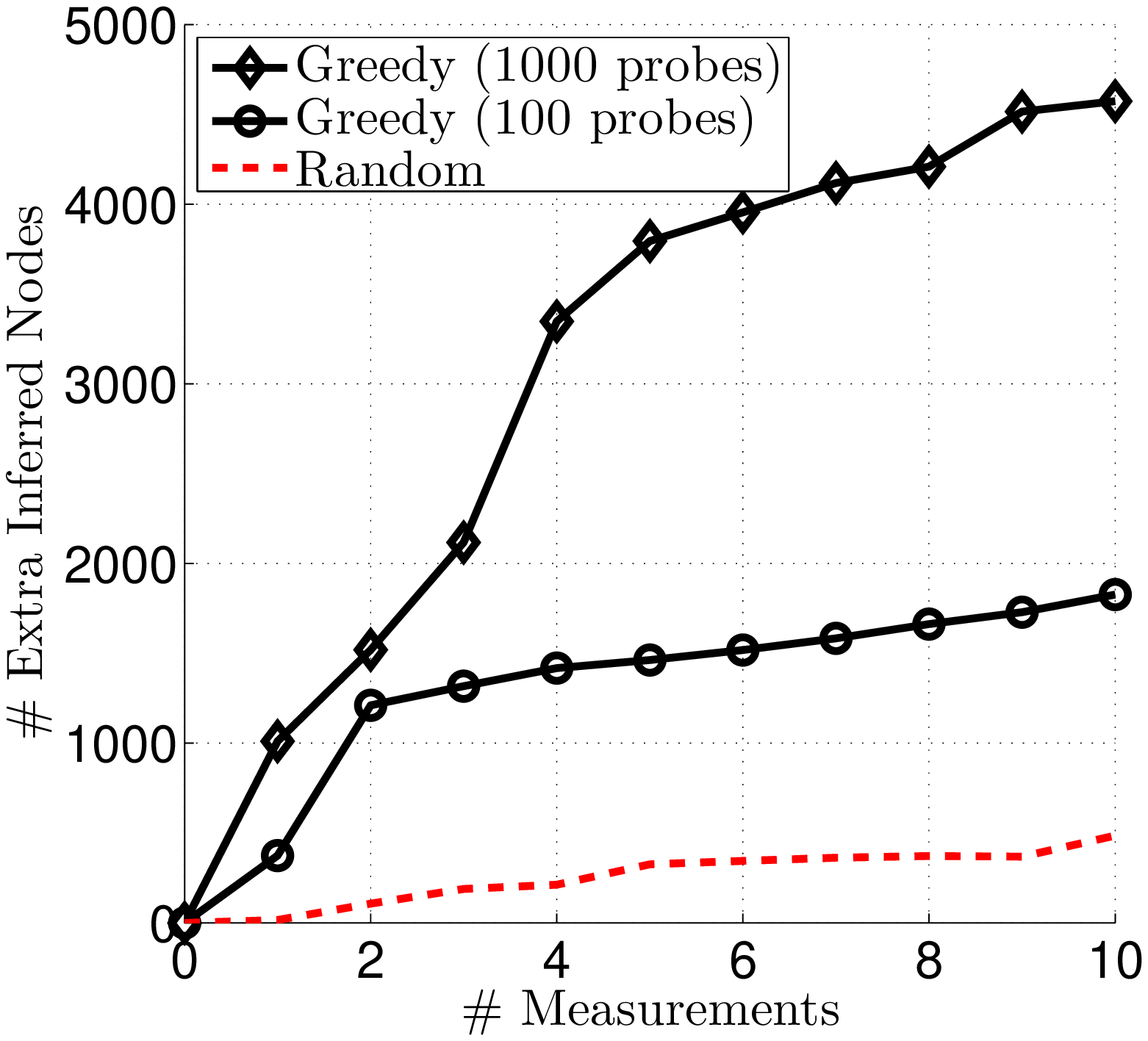}\label{fig:greedy-high}}
\vspace{-5mm}
\caption{Number of extra nodes whose routes can be inferred with certainty (y-axis) when a set of measurements $\mathcal{X}$ is provided ($|\mathcal{X}|$ in x-axis), in two scenarios with (a) low and (b) high number of nodes with initially uncertain route.}
\label{fig:greedy}
\end{figure}

\subsection{\camera{Real-World Evaluation}}\label{sec:real-evaluation}

\camera{Besides simulations, here, we provide evaluation results from measurements and experiments in the real Internet.}

\myitem{\camera{Measurements for MOAS prefixes.}}
\camerax{
The proposed inference framework can be applied on top of any given topology and routing model (\ie it takes this information as input). As a result, its accuracy
}
\camera{
depends on how \textit{complete} and \textit{accurate} the knowledge of (i) the routing policies $\mathcal{Q}$ and $\mathcal{H}$ and (ii) the AS-level graph is (perfect knowledge leads to $100\%$ inference accuracy).} 

\camera{We verified this by comparing our inference results against real BGP routing entries collected from more than $200$ route collectors of RIPE RIS~\cite{riperis} and RouteViews~\cite{routeviews} for around $300$ 
prefixes that are anycasted by more than one AS (\ie Multi-Origin AS, MOAS) in the Internet~\cite{caidapfx2as}. When using the VF model (see Section~\ref{sec:vf-model}) and the available AS-relationships~\cite{caidaasrel}, the achieved accuracy (for networks whose routing entries are available) is 60-70\%\footnote{An interesting observation is that assuming shortest path preference (Section~\ref{sec:weak-path-inference}) increases the inference completeness from  
30\% to 65\%, without significantly affecting the accuracy;
this supports the real-world relevance of this assumption~\cite{anwar2015investigating, gill2011let,c-bgp}.}, which complies with the observed accuracy of the VF-model for Internet routing~\cite{anwar2015investigating}. As a comparison, the accuracy of simulation-based catchment prediction in these scenarios is $10\%$ lower than the accuracy of the certain inference with shortest path preference.}
\camera{Introducing fine-grained refinements in the routing policies for some nodes, increases accuracy; we tested this by considering per-prefix policies, similarly to~\cite{anwar2015investigating}. 
Specifically, we re-ran our inference by ``correcting'' (\ie replacing, adding, removing) in the \Rgraph the links close to the anycasters (starting from the first hops), with the actually observed links in the measured paths
. 
}
\camerax{For example, if a link was not included in the initial topology dataset (AS-relationships~\cite{caidaasrel}), but was observed in the real measurements, we added it in the \Rgraph to increase the completeness of the topology; or, if an observed link existed in the topology, but did not appear in the \Rgraph (\ie due to an inaccurate routing policy), we similarly added it in the ``corrected'' \Rgraph. }
\camera{With a 30\% of the links observed by the monitors being corrected, the average accuracy increases to 80\%. This observation validates that the inference accuracy depends on the underlying knowledge of the topology and routing policies.}

\camerax{Moreover, the structure of the \Rgraph provides further insights about what are the important links and policies for a routing configuration, and how missing information (\eg topology incompleteness) would affect inference. For example, a link that is in the topology dataset but does \textit{not} appear in the \Rgraph, does not affect the inference, \ie missing this link (\eg in an incomplete dataset) would not be important. Similarly, a link that appears in the \Rgraph but removing it does not affect the \textit{certain} inference of any node
, would not be important for the examined routing configuration. This information could be used --similarly to our experiments-- to design methods for targeted corrections (\eg through targeted measurements) of a topology/routing model.}


\myitem{\camera{Anycast experiments with the PEERING testbed.}}
\camera{We conducted controlled IP anycast experiments in the real Internet using the PEERING testbed~\cite{schlinker2014peering,peering}, which owns ASNs, IP prefixes and has BGP connections with operational networks in several locations around the world. We announce the same prefix from different PEERING locations, \ie \ppoints. Figure~\ref{fig:peering-anycast-example} shows the fraction of the catchment of each \ppoint in four experiment scenarios (SC-0, SC-1, SC-2, and SC-2*), as measured from the RIPE RIS~\cite{riperis} and RouteViews~\cite{routeviews} route collectors 
(black bars), and inferred using our framework (\eq{eq:traffic-load} with $T_{i}=T, \forall i$) on top of the VF model (white bars).}

\camera{We consider an initial scenario (``SC-0''), where a network has two \ppoints, \textit{AMS} and \textit{UFMG} (
for more details about the PEERING locations see~\cite{peering}).  As seen in Fig.~\ref{fig:peering-anycast-example}, the load distribution is highly skewed towards \textit{AMS}; our inference captures this imbalance, with a 10\% deviation from measured values. The network would like to evaluate whether it can balance the load by adding more \ppoints, before proceeding to an actual deployment (\ie hypothetical scenario).}
    
\camera{One option (``SC-1'') is to add an extra \ppoint (at the PEERING location \textit{GRNET}). However, our inference (white bars for SC-1) predicts that this would have only a small effect on the load distribution, and thus it would be an inefficient deployment. Our experimental results (black bars for SC-1) verify this behavior, i.e., the added \ppoint in SC-1 (``other'' bars) attracts a small percentage of traffic
. Hence, the network considers a second option (``SC-2'') to add two other \ppoints (at the PEERING locations \textit{ISI} and \textit{UW}). Our inference predicts that SC-2 would (i) move a significant fraction of load from \textit{AMS} to the added \ppoints, and (ii) not affect the load of \textit{UFMG}, \ie SC-2 achieves a better load balancing than SC-1. While deviations between inferred and measured catchment exist also here (due to the employed naive VF model
), the actual behavior is approximated well by our predictions.}

\camera{Moreover, we calculated the certain catchment \textit{without} shortest path preference for \textit{AMS} in SC-2, which corresponds to a ``lower bound'' (cf. Section~\ref{sec:rgraph-vs-sims}) for the \textit{AMS} catchment, i.e., under any path length combination. We found the \textit{AMS} certain catchment to be almost zero (not shown in Fig.~\ref{fig:peering-anycast-example}). This means that the short path lengths towards \textit{AMS} are the main causes for traffic to be routed to this \ppoint. Therefore, prepending the announcements from \textit{AMS} (to artificially increase path lengths) could further decrease the attracted load. In fact, since the \textit{AMS} certain catchment is very small, an intensive prepending could even diminish the load in \textit{AMS}. We verified this through experiments (``SC-2*'') with the \ppoints of SC-2, where we prepended 5 hops (i.e., more than the median AS-path length in the Internet~\cite{sermpezis2017can}) in the announcements from \textit{AMS}.}

\begin{figure}
\centering
\subfigure[SC-0]{\includegraphics[width=0.22\linewidth]{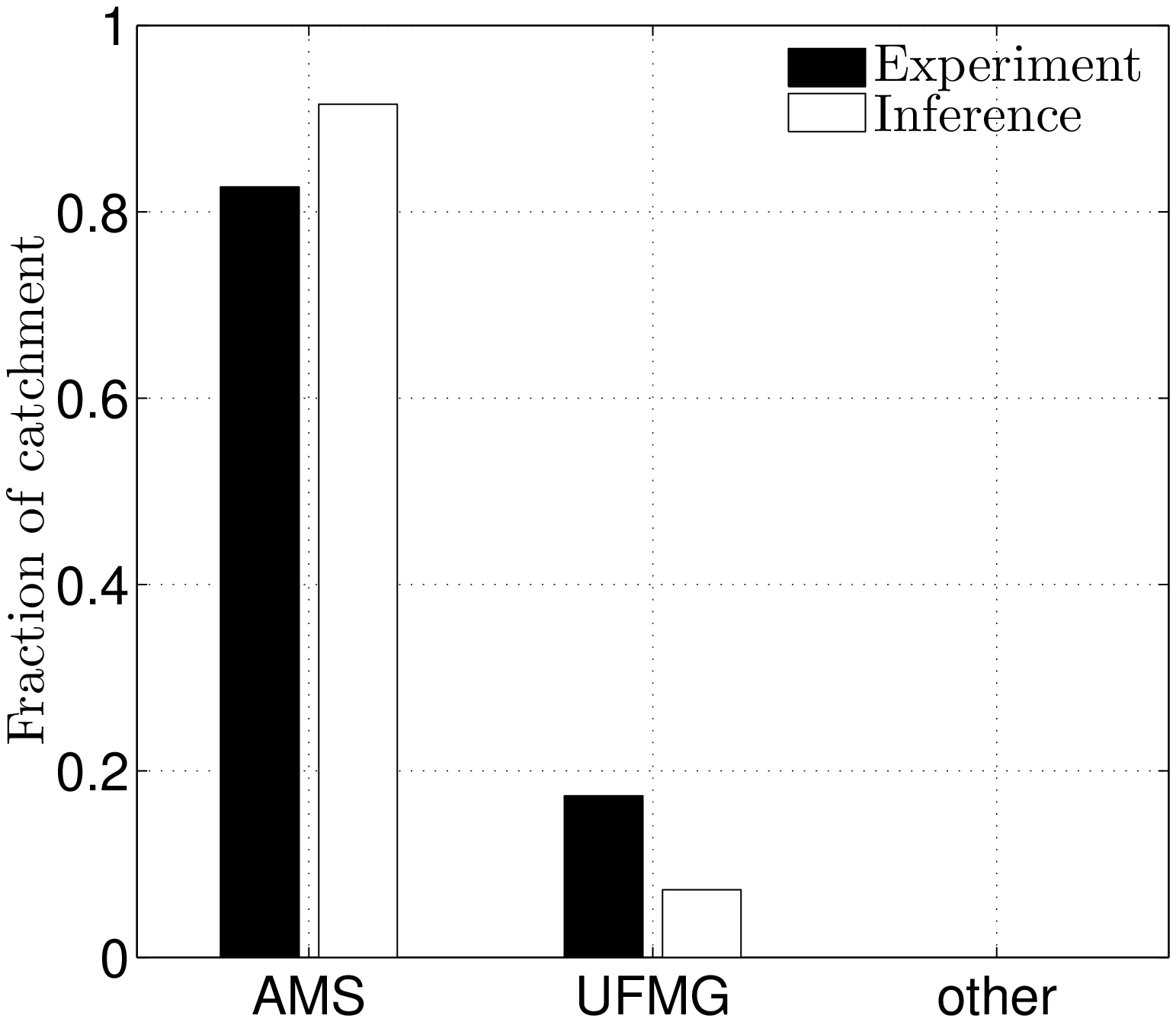}\label{fig:peering-anycast-example-0}}
\hspace{0.02\linewidth}
\subfigure[SC-1]{\includegraphics[width=0.22\linewidth]{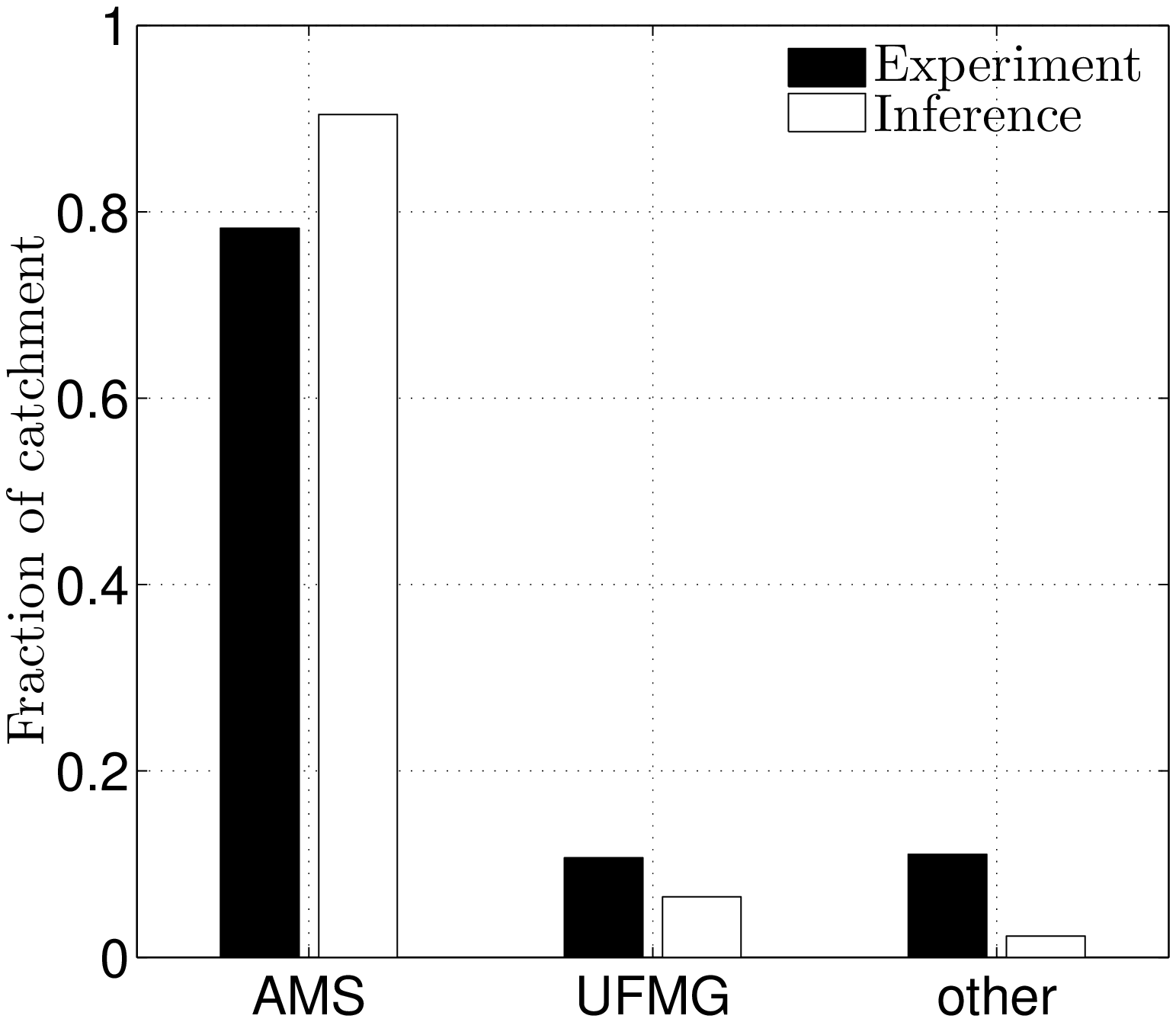}\label{fig:peering-anycast-example-1}}
\hspace{0.02\linewidth}
\subfigure[SC-2]{\includegraphics[width=0.22\linewidth]{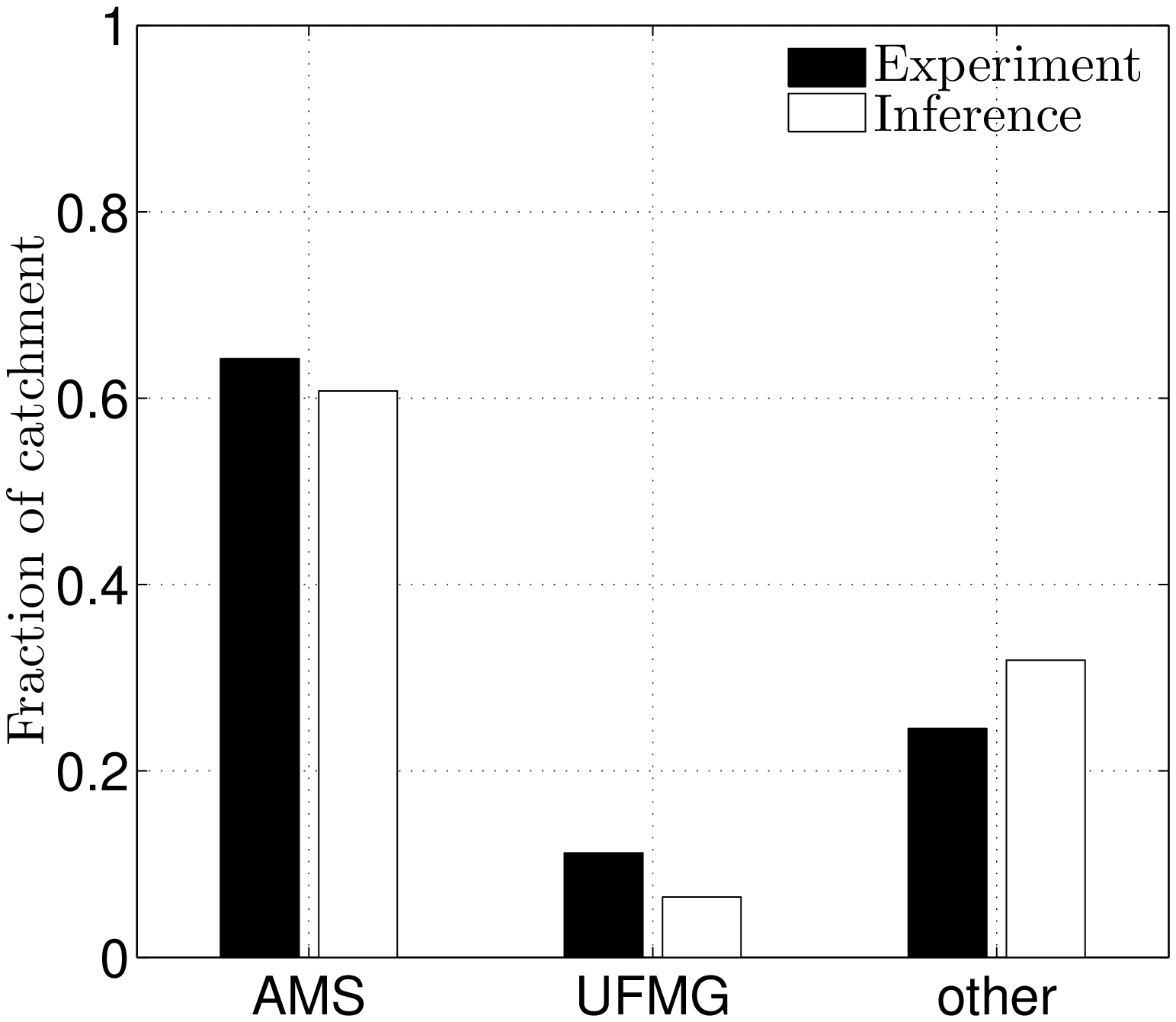}\label{fig:peering-anycast-example-2}}
\hspace{0.02\linewidth}
\subfigure[SC-2*]{\includegraphics[width=0.22\linewidth]{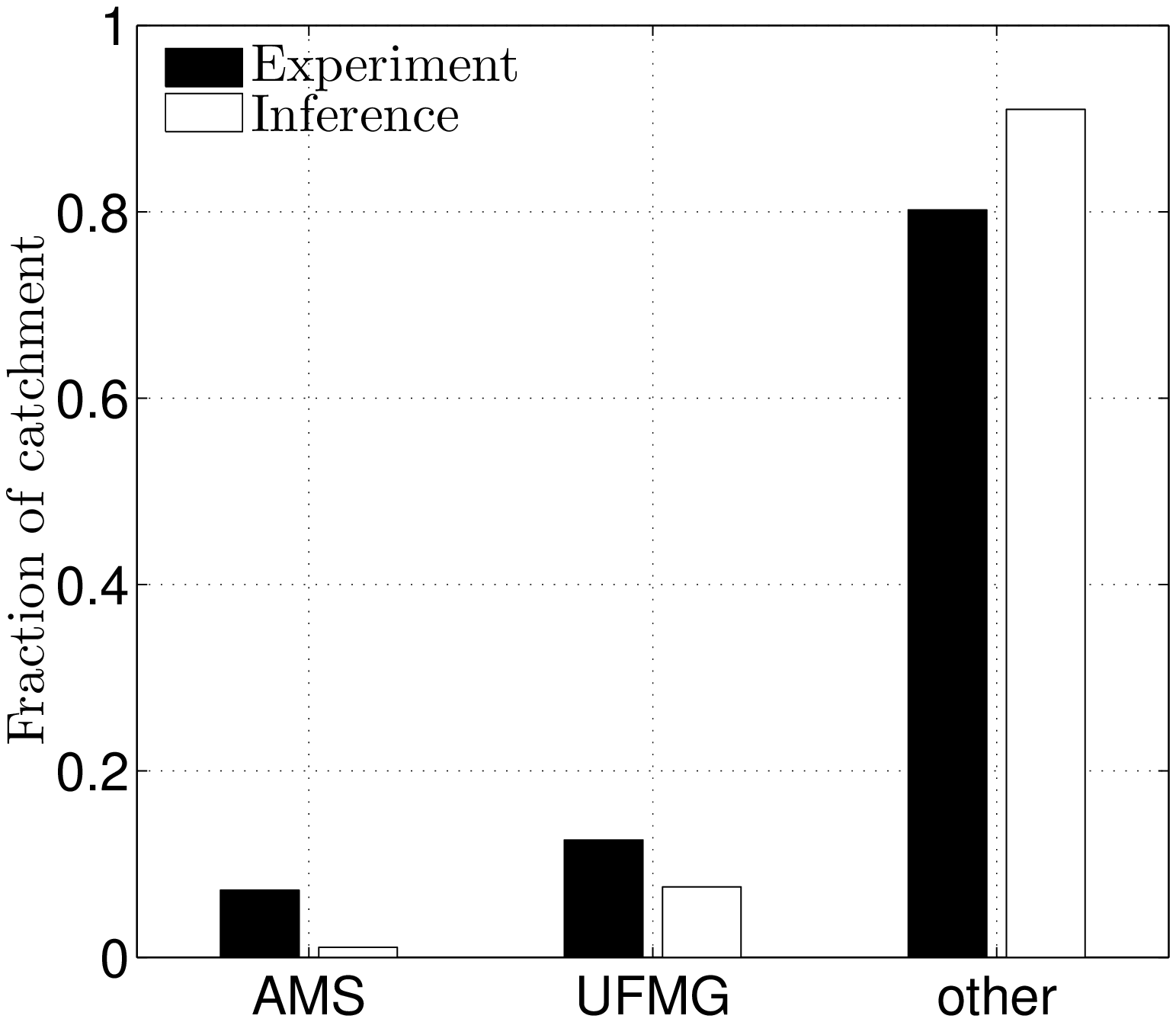}\label{fig:peering-anycast-example-2p}}
\caption{\camera{Experiments in the real Internet with the PEERING testbed: Catchment of the \ppoints \{\textit{AMS, UFMG, other}\} in four deployments/scenarios (SC-0, SC-1, SC-2, SC-2*) as observed from monitors in the Internet (black bars) and predicted using our framework (white bars)}.}
\label{fig:peering-anycast-example}
\end{figure}

\section{Related Work}\label{sec:related}

The majority of related literature focuses on methodologies for \emph{measuring} the catchment in existing deployments~\cite{baltra2014ingress,mao2005level,lee2011scalable,cicalese2015fistful,de2017broad,ra2017seeing}. A methodology for measuring the routes towards the different \ppoints of a destination network, based on past measurements, is proposed in~\cite{baltra2014ingress}. Similarly, \cite{mao2005level} infers AS-level paths $\bestaspath{i}{n_{dst}}$ without measurements from the source network $i$, but based on BGP tables collected from multiple vantage points, AS-relationship data, and valley-free assumptions, while~\cite{lee2011scalable} infers routes by stitching path segments from existing measurements. Latency-based~\cite{cicalese2015fistful} and data-plane~\cite{de2017broad} measurement methodologies have been recently proposed for mapping anycast catchment. In \textit{Verfploeter}~\cite{de2017broad}, a system in the network of $n_{dst}$ performs exhaustive ping measurements (to all routed IP prefixes) and monitors from which \ppoint the reply packets arrive to $n_{dst}$. 
\blue{In contrast to these works, our methodology can infer catchment also on \textit{hypothetical} deployments \camera{(see, \eg Section~\ref{sec:real-evaluation})}, a task more challenging than the already demanding task of calculating existing catchment~\cite{rfc4786,lodhi2015complexities}. Furthermore, our methodology can complement existing measurement methods, and can be used as a base for devising more lightweight/efficient techniques, \eg by exploiting the structure and knowledge offered by the \Rgraph (similarly and by extending the concepts presented in Section~\ref{sec:measurements}).
}




\blue{Prior work on Internet route prediction comprises mainly simula\-tion-based approaches~\cite{gao2001stable,c-bgp,muhlbauer2006building, feamster2004model}, which simulate the operation of BGP based on known or estimated routing policies. Our work builds on top of these approaches, and provides more informative results and insights. 
}
\blue{The works~\cite{c-bgp,feamster2004model} develop models~\cite{feamster2004model} and tools~\cite{c-bgp,feamster2004model} mainly for the intra-domain routing (iBGP) and traffic engineering (TE) of egress traffic, whereas our goal is to predict ingress routes and perform TE with eBGP policies. Nevertheless, these approaches could be combined with ours for a joint intra/inter-domain TE. The importance of the intra-domain structure to the inter-domain routing is highlighted in~\cite{muhlbauer2006building}, whose approach is orthogonal and could be used complementarily to our approach as well, \eg to provide more fine-grained routing policies $\mathcal{Q}$ and $\mathcal{H}$.
}

Route inference or prediction has been employed in different contexts as well, \eg for designing targeted active measurements~\cite{cunha2016sibyl}, optimal monitor placement~\cite{gregori2012incompleteness}, or investigation of potential path redundancy in the Internet~\cite{kloti2015policy}. Our framework can be used complementary to these works. 
\camera{Finally, probabilistic network programming languages~\cite{gehr2018bayonet}, which capture probabilistic network behavior and analyze it through standard probabilistic inference methods, could be combined with our work to design novel efficient inference tools for Internet routing applications.}

\section{Conclusions}\label{sec:conclusions-future-work}
We proposed and studied a methodology to infer routing behavior in the Internet for existing or hypothetical topological and routing configurations. Our methodology deviates from and enhances existing approaches, by predicting \ppoint catchment with certainty or probabilistically, with or without measurements, and under generic routing assumptions. 
\camera{Our methods can be useful for a number of network management application, as well as open new research directions; some indicative examples are:}



\myitem{\camera{Applications.}}
\camera{
(i) \textit{Traffic Engineering:} An operator can efficiently predict and obtain rich information (lower/upper bounds through certain catchment, effect of path lengths, etc.) about the impact of adding/removing \ppoints or doing path prepending; due to the large number of possible actions and their combinations, evaluation through experiments or simulation-based approaches would become inefficient. (ii) \textit{Peering strategy:} Establishing a new peering connection with a single network or many networks at an IXP, may significantly change the catchment of \ppoints or may have negligible impact (see experiments in Section~\ref{sec:real-evaluation}). Today, networks have higher flexibility in establishing peerings even with distant networks, \eg through resellers and remote peering~\cite{castro2014remote,nomikos2018peer}; catchment prediction can enable them to make informed decisions, before proceeding to actual deployments. (iii) \textit{Resilience:} Our framework facilitates to study the resilience of a network against failures of \ppoints or peering links. The structure and properties of the \Rgraph (\eg centrality) can further reveal the links whose failure would affect the network the most.
We believe that a graph-theoretic approach, based on the \Rgraph, could complement and enhance existing measurement-based approaches, \eg~\cite{fontugne2018thin}. (iv) \textit{Network security:} IP anycast is used by DDoS protection organizations to attract and scrub DDoS traffic destined to a victim network~\cite{de2016anycast}, or to mitigate hijacking attacks~\cite{sermpezis2018artemis}. These organizations can select where to deploy \ppoints in order to maximize their catchment (\eg by mapping potential attackers to ``illegitimate'' \ppoints), and thus best protect their customers.
}

\myitem{\camera{Future research directions.}}
\camera{
We identify two future research directions that could be facilitated by our framework. (i) \textit{Internet routing models:} Existing models for Internet topology and routing~\cite{gao2001stable,c-bgp,muhlbauer2006building,feamster2004model}, such as the AS-graph and the VF, are widely used in research and network operations. Despite (or, due to) their generality, they suffer from limited accuracy. However, this accuracy can be increased when modeling the topology and the policies from the perspective of a single network (\eg per-prefix policies; see~\cite{anwar2015investigating} and Section~\ref{sec:real-evaluation}), rather than having a common topology/routing model for all networks. To this end, one could use the \Rgraph, which encodes the topology and routing policies from the perspective of a single network, $n_{dst}$. For example, a \Rgraph can be created on top of a general model, and then refined from real measurement data, \eg similarly to~\cite{muhlbauer2006building}. \camerax{While building a \textit{general} (\ie for all networks) data-driven model, such as~\cite{muhlbauer2006building}, may require a very large number of measurements~\cite{gregori2012incompleteness} to capture accurately all routing information, when it comes to the perspective of a single network $n_{dst}$, one can focus her efforts only on the ``important'' links (see Section~\ref{sec:real-evaluation}).} (ii) \textit{Reinforcement learning for network management:} Optimization problems that arise in network management processes are frequently combinatorial (see, \eg Section~\ref{sec:measurements}), and thus difficult to solve with analytic methods. Recently, Reinforcement Learning (RL) methods have been proposed for efficient network management and routing operations~\cite{yu2018drom, yao2018networkai}.
In absence of real data, RL agents can be trained on simulated environments. Our framework offers richer information than simulations and requires less computations (\eg see Section~\ref{sec:rgraph-vs-sims}). Hence, it could significantly reduce the time needed for the training of a RL agent that would involve testing over a large number of different scenarios.
}

\camera{To facilitate further research and reproducibility, we make the code for an implementation of the proposed methods available in~\cite{sermpezis2019code}.}

\section*{APPENDIX}
\appendix
\section{Proof of Theorem~\ref{thm:vf-graph}}\label{appendix:proof-theorem-vf-graph}
We first define as \Rpath from $i$ to $n_{dst}$, a path created by starting at $n_{dst}$ and following directed edges until reaching $i$.

We prove the Theorem, by proving the following two items: (i) any path in the \Rgraph (\ie \Rpath), is an eligible path; (ii) any eligible path is encoded in the \Rgraph as a \Rpath.

\myitemit{Any path in the \Rgraph (\ie \Rpath), is an eligible path.} Let a \Rpath $rp = [n_{1}, n_{2}, ..., n_{K}, n_{dst}]$. The $rp$ is constructed by following edges in the \Rgraph, which means that $e_{n_{k+1}n_{k}}\in\mathcal{E_{R}}$, where $k=1,...,K-1$. The existence of the edge denotes that (see Algorithm~\ref{alg:build-vf-graph}): (a) $n_{k+1}$ is in the set \textit{best\_neighbors} of $n_{k}$, or equivalently $q_{n_{k}n_{k+1}}\geq q_{n_{k}j}$, $\forall j\in \{i\in \mathcal{N_{R}}: e_{in_{k}}\in\mathcal{E_{R}}\}$, and thus can be selected by $n_{k}$. (b) $n_{k+1}$ exports its best path $\bestaspath{n_{k+1}}{n_{dst}}$ to $n_{k}$, and routes all paths of equal local preference similarly (see Section~\ref{sec:generic-model}); thus any path $[n_{k+1}, x, ..., n_{K}, n_{dst}]$ can be a path that reaches the RIB of $n_{k}$, for any $x$ that $q_{n_{k+1}x} = q_{n_{k+1}n_{k+2}}$. These two conditions satisfy the definition of eligible paths (Def.~\ref{def:eligible-path}, Section~\ref{sec:model}).

\myitemit{Any eligible path is encoded in the \Rgraph as a \Rpath.} Let 
an eligible path $ep = [n_{1}, n_{2}, ...\\n_{K}, n_{dst}]$ that is not a \Rpath, \ie at least one edge in $ep$ does not exist in the \Rgraph; let this edge be between $n_{k+1}$ and $n_{k}$. Let also a node $x$ that is a parent of $n_{k}$ in the \Rgraph (\ie $e_{xn_{k}}\in\mathcal{E_{R}}$). Then, it must hold that $q_{n_{k}x}>q_{n_{k}n_{k+1}}$ or $h_{n_{k+1}n_{k+2}n_{k}}=0$. In the former case, the path $[n_{k},n_{k+1},...,n_{dst}]$ cannot be the best path of $n_{k}$ and thus $n_{1}$ will never have in its RIB the path $ep$ (contradiction). In the latter case, the path $[n_{k+1},n_{k+2},...,n_{dst}]$ is never exported to $n_{k}$, which means that $ep$ does not conform to routing policies (contradiction).

\section{Proof of Lemma~\ref{thm:update-prob-np-hard}}\label{appendix:thm-update-prob-np-hard}
In general, a node $i$ in the \Rgraph has more than one paths to $n_{dst}$, which means that the \Rgraph is a multiply-connected BN (and not a \textit{polytree}). The problem of updating the probabilities (or, ``belief updating'') in non-polytree BNs is known to be NP-hard (by reduction to a SAT problem)~\cite{cooper1990computational}.

\section{Proof of Theorem~\ref{thm:efficiency-measurement-algo}}\label{appendix:thm-efficiency-measurement-algo}
\textit{Correctness:} A node $i$ has a certain route only if (a) all its parents $P_{i}$ have a certain route, or (b) (at least) one of its children $j\in C_{i}$ has a certain route and $j$ routes through $i$. The former case is captured by the condition in \textit{line 23} (for node $j$ and its parents), and the latter in \textit{line 13} where the condition requires that $i$ routes traffic through the node $j\equiv CP_{i}$.

\noindent\textit{Completeness:} The updating process of Algorithm~\ref{alg:coloring-from-measurements} is based on the fact that a BN node is conditionally independent of all other nodes given its Markov blanket, \ie its parents, children, and ``spouses'' (parents of common children)~\cite{pearl1988prob-reasoning}. Hence, for each oracle, let for a node $i$, Algorithm~\ref{alg:coloring-from-measurements} visits \textit{all nodes that are dependent} on $i$, \ie its parents (\textit{line 12}), children (\textit{line 16}), and parents of children (\textit{line 18}). Any other node is not dependent, unless a change in the value/route of a visited node takes place (in that case Algorithm~\ref{alg:coloring-from-measurements} calls again \textproc{SetRoute} for this node; in \textit{line 14} or \textit{24}).

\noindent\textit{Complexity:} The function \textproc{SetRoute} is called only for nodes with an oracle (\textit{line 3}), or only for nodes without a certain route (see in \textit{line 13} condition $f(CP_{i}=0)$, and in \textit{line 24}, node $j\in C_{i}$ satisfies the condition $f(j)=0$ from \textit{line 16}). As soon as a node sees a certain route, it is not considered for further inference. Let $i,j\in\mathcal{X}$ and $k$ a neighbor of both $i$ and $j$; if \textproc{SetRoute} is called for $k$ when $i$ is visited, then it will not be re-called for $k$ when $j$ is visited. Hence, \textproc{SetRoute} is called at most $|\mathcal{N}|$ times.
\textit{Remark: }This does not mean that the recursion depth of \textit{lines 14} and \textit{24} is at most one, but that \emph{the sum of recursive calls of \textproc{SetRoute} is bounded by} $|\mathcal{N}|$.

\section{Proof of Theorem~\ref{thm:weak-paths-increase-coloring}}\label{appendix:thm-weak-paths-increase-coloring} 
We provide a sketch of the proof. By design, Algorithm~\ref{alg:transform-vf-bn-weak} only removes edges from the \Rgraph. A node $i$ has a certain route only if all its parents $P_{i}$ have the same certain route as well; removing a parent changes neither the route of the other parents, nor the route of $i$. On the other hand, let all parents of $i$, except for one parent $j\in P_{i}$, have the same route $m$; then $i$ does not have a certain route. If the edge $e_{ji}$ is removed, the remaining parents of $i$ will have the same route $m$, and thus a new route inference for node $i$ can be safely made.

\section{Proof of Lemma~\ref{lemma:properties-problem-certain}}
\label{appendix:thm-properties-modularity}

The first item follows straightforwardly from the definition of the function $|\mathcal{NC}_{R}(\mathcal{X})|$ (the size of a set is non-negative), and the fact that a measurement is only an observation that cannot change the (certainly inferred) route of a node and thus decrease the number of nodes which already have certain routes. In particular, a certain route for node $i$ is independent of the route probabilities of other nodes without a certain inference (otherwise the route of $i$ would not have been inferred with certainty). Hence, (a) in case the measured node already has a certain route, a (valid) measurement cannot change this route, and (b) in case the measured node does not have a certain route, then it does not affect the route of $i$; in either case the route of $i$ is not affected.

We prove the second and third items through two counter-examples depicted in Fig.~\ref{fig:examples-submodularity}. 

\begin{figure}
\centering
\includegraphics[width=1\linewidth]{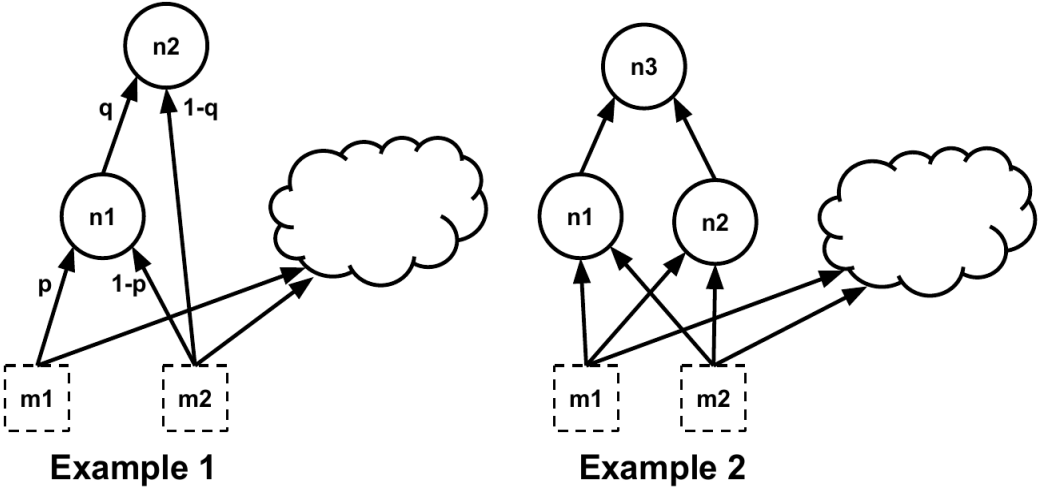}
\caption{Examples of a \Rgraph. The \ppoints $m1$ and $m2$ are denoted with boxes, the nodes of interest are denoted with circles, and the clouds correspond to a part of the \Rgraph that has no (outgoing) edges towards the nodes of interest and whose structure is not of interest. The weight next to each edge denotes the probability for a node to select the route from this incoming edge.}
\label{fig:examples-submodularity}
\end{figure}

We remind that a set function $f$ is \textit{submodular} when $\forall A\subseteq B $ and $\forall \epsilon\notin A$ it holds that
\begin{equation}\label{eq:submodularity-def}
f(A\cup \{\epsilon\}) - f(A) \geq f(B\cup \{\epsilon\}) - f(B)
\end{equation}
and is \textit{supermodular} when $\forall A\subseteq B $ and $\forall \epsilon\notin A$ it holds that
\begin{equation}\label{eq:supermodularity-def}
f(A\cup \{\epsilon\}) - f(A) \leq f(B\cup \{\epsilon\}) - f(B)
\end{equation}
In other words, the marginal gain in a submodular function by adding an element $\epsilon$ to a set $S$ diminishes with the size of the set $S$.

\myitem{Example 1.} Consider the first example in Fig.~\ref{fig:examples-submodularity}, let $A$ be a set of nodes in the ``cloud'' of Fig.~\ref{fig:examples-submodularity}, and let
\begin{align*}
A & \cap \{n1,n2\}    = \emptyset \\
B &\equiv A \cup \{n2\} \\
\epsilon & \equiv \{n1\}
\end{align*}
In this case, the objective function of \eq{eq:objective-function-greedy} takes the value
\begin{equation}
E_{P}\left[\mathcal{NC}_{R}(A\cup \{\epsilon\})\right] = E_{P}\left[\mathcal{NC}_{R}(A)\right] + 1 + (1-p)
\end{equation}
because we will have an oracle for $n1$ (i.e., we increment by 1), and if this oracle is $n1\rhd m2$ (which happens with probability $1-p$) then we can certainly infer that $n2$ routes to $m2$ as well (i.e., we increment one more by 1, but now with probability $1-p$); otherwise we cannot infer the route of $n2$ with certainty (and we do not increment).

Also for $B$ we get for the objective function
\begin{equation}
E_{P}\left[\mathcal{NC}_{R}(B)\right] = E_{P}\left[\mathcal{NC}_{R}(A)\right] + 1 + p\cdot q
\end{equation}
because we will have an oracle for $n2$ (i.e., we increment by 1), and if this oracle is $n2\rhd m1$, this means that we can also infer that $n1\rhd m1$, because this is the only way that $n2$ can route to $m1$, and the respective probability is $p\cdot q$ (\ie w.p. $p$ the node $n1$ routes to $m1$ \textit{and} $n2$ selects the route from $n1$ w.p. $q$); otherwise (\ie $n1\rhd m2$) we cannot infer the route of $n1$ with certainty.

Finally, trivially, we get
\begin{equation}
E_{P}\left[\mathcal{NC}_{R}(B\cup \{\epsilon\})\right] = E_{P}\left[\mathcal{NC}_{R}(A)\right] + 2
\end{equation}
because we have oracles for both nodes $n1$ and $n2$.

The above equations give 
\begin{align}
\Delta_{A} &= E_{P}\left[\mathcal{NC}_{R}(A\cup \{\epsilon\})\right] - E_{P}\left[\mathcal{NC}_{R}(A)\right] = 1 + (1-p) \\
\Delta_{B} &= E_{P}\left[\mathcal{NC}_{R}(B\cup \{\epsilon\})\right] - E_{P}\left[\mathcal{NC}_{R}(B)\right] = 1- p\cdot q
\end{align}
It is easy to see that $\Delta_{A}\geq 1 \geq \Delta_{B}$, which means that \textit{the objective function cannot be supermodular}, because there exists an $\epsilon$ for which the inequality \eq{eq:supermodularity-def} does not hold.

\myitem{Example 2.} Consider the second example in Fig.~\ref{fig:examples-submodularity}, let $A$ be a set of nodes in the ``cloud'' of Fig.~\ref{fig:examples-submodularity}, and let
\begin{align*}
A & \cap \{n1,n2,n3\}    = \emptyset \\
B &\equiv A \cup \{n2\} \\
\epsilon & \equiv \{n1\}
\end{align*}
The objective function of \eq{eq:objective-function-greedy} takes the value 
\begin{equation}
E_{P}\left[\mathcal{NC}_{R}(A\cup \{\epsilon\})\right] = E_{P}\left[\mathcal{NC}_{R}(A)\right] + 1
\end{equation}
because we will have an oracle for $n1$ (i.e., we increment by 1), and no matter what this oracle is, the route probabilities $\pi$ for $n2$ and $n3$ will be always non-zero for both $m1$ and $m2$ (\ie we cannot make any other certain inference).

Also for $B$ we get for the objective function
\begin{equation}
E_{P}\left[\mathcal{NC}_{R}(B)\right] = E_{P}\left[\mathcal{NC}_{R}(A)\right] + 1
\end{equation}
because we will have an oracle for $n2$ (i.e., we increment by 1), and no matter what this oracle is, the route probabilities $\pi$ for $n1$ and $n3$ will be always non-zero for both $m1$ and $m2$ (\ie we cannot make any other certain inference).

Finally, let $w$ denote the probability that $n1$ and $n2$ route to the same \ppoint. Then, we get
\begin{equation}
E_{P}\left[\mathcal{NC}_{R}(B\cup \{\epsilon\})\right] = E_{P}\left[\mathcal{NC}_{R}(A)\right] + 2 + w
\end{equation}
because we will have an oracle for nodes $n1$ and $n2$ (i.e., we increment by 2), and if both oracles for $n1$ and $n2$ are for the same \ppoint (which happens w.p. $w$), we can make one more certain inference for $n3$; otherwise we cannot infer the route of $n3$ with certainty.

The above equations give 
\begin{align}
\Delta_{A} &= E_{P}\left[\mathcal{NC}_{R}(A\cup \{\epsilon\})\right] - E_{P}\left[\mathcal{NC}_{R}(A)\right] = 1 \\
\Delta_{B} &= E_{P}\left[\mathcal{NC}_{R}(B\cup \{\epsilon\})\right] - E_{P}\left[\mathcal{NC}_{R}(B)\right] = 1 + w
\end{align}
It is easy to see that $\Delta_{A} \leq \Delta_{B}$, which means that \textit{the objective function cannot be submodular}, because there exists an $\epsilon$ for which the inequality \eq{eq:submodularity-def} does not hold.

\end{document}